\pdfoutput=1
\RequirePackage{amsmath} % put here so that llncs can redefine some things
\documentclass[oribibl, envcountsame, 11pt]{llncs}
\usepackage[utf8]{inputenc}

\spnewtheorem*{theorem*}{Theorem}{\normalfont\bfseries}{\itshape}

\usepackage{fullpage}
\usepackage{anyfontsize}
\usepackage{breqn}

\usepackage{mathtools}

\usepackage[numbers,sectionbib]{natbib}
\usepackage{amsfonts}
\usepackage{mathrsfs}
\usepackage{amssymb}
\usepackage{dsfont}
\usepackage{xfrac}

\usepackage{tikz}
\usetikzlibrary{decorations.pathmorphing,decorations.pathreplacing}

\usepackage[colorlinks=true,urlcolor=webblue,linkcolor=webgreen,filecolor=webblue,citecolor=webgreen,pdfpagemode=UseOutlines,pdfstartview=FitH,pdfpagelayout=OneColumn,bookmarks=true]{hyperref}
\definecolor{webgreen}{rgb}{0,.5,0}
\definecolor{webblue}{rgb}{0,0,.5}

\newcommand{\ket}[1]{\vert{#1}\rangle}
\newcommand{\bra}[1]{\langle{#1}\vert}
\newcommand{\proj}[1]{\ket{#1}\bra{#1}}
\newcommand{\ip}[2]{\bra{#1}#2\rangle}
\newcommand{\abs}[1]{\left|#1\right|}	
\newcommand{\tr}{\text{\textnormal{tr}}}
\newcommand{\Tr}{\text{\textnormal{Tr}}}
\newcommand{\reg}[1]{{\color{gray} #1}}
\newcommand{\F}{\mathcal{F}}
\newcommand{\SP}[2]{(#1,#2)_{\mathsf{Sp}}}
\usepackage{scalerel}
\DeclareMathOperator*{\expectation}{\scalerel*{\mathbb{E}}{\textstyle\sum}}

\DeclarePairedDelimiter{\dnorm}{\|}{\|_{\diamond}}
\DeclarePairedDelimiter{\trnorm}{\|}{\|_{\mathrm{tr}}}

\newcommand{\id}{\mathbb{I}}
\newcommand{\I}{\mathsf{I}}

\renewcommand{\H}{\mathsf{H}}

\newcommand{\X}{\mathsf{X}}
\newcommand{\Y}{\mathsf{Y}}
\newcommand{\Z}{\mathsf{Z}}
\newcommand{\Pauli}{\mathbb{P}}
\renewcommand{\epsilon}{\varepsilon}

\newcommand{\KeyGen}{\mathsf{KeyGen}}
\newcommand{\Key}{\mathcal{K}}

\newcommand{\Encrypt}{\mathsf{Encrypt}}
\newcommand{\Decrypt}{\mathsf{Decrypt}}
\newcommand{\Real}{\mu^{\mathsf{real}}}
\newcommand{\Ideal}{\mu^{\mathsf{ideal}}}

\newcommand{\advA}{\mathcal{A}}
\newcommand{\cacc}{\ensuremath{\mathsf{acc}}}
\newcommand{\crej}{\ensuremath{\mathsf{rej}}}
\newcommand{\simacc}{\mathcal{S}_{\cacc}}
\newcommand{\simrej}{\mathcal{S}_{\crej}}
\newcommand{\simS}{\mathcal{S}}
\newcommand{\negl}{\operatorname{negl}}

\newcommand{\real}{\ensuremath{\mathsf{real}}}
\newcommand{\ideal}{\ensuremath{\mathsf{ideal}}}

\newcommand{\QCA}{\mathsf{QCA}}
\newcommand{\QCAR}{\operatorname{\mathsf{QCA-R}}}
\newcommand{\DNS}{\mathsf{DNS}}
\newcommand{\GYZ}{\mathsf{GYZ}}

\usepackage{algorithm}
\usepackage[noend]{algpseudocode}
\makeatletter
\renewcommand\fnum@algorithm{\fname@algorithm~\thealgorithm.}
\makeatother
\makeatletter
\newenvironment{construction}[1]{%
    \renewcommand{\ALG@name}{Construction}% Update algorithm name
   \begin{algorithm}#1%
  }{\end{algorithm}}
\makeatother

\title{Quantum ciphertext authentication and key recycling with the trap code}

\author{Yfke Dulek\inst{1} \and Florian Speelman\inst{2}}
\institute{
Centrum voor Wiskunde en Informatica (CWI), QuSoft, and University of Amsterdam
\and
QMATH, Department of Mathematical Sciences, University of Copenhagen}

\begin{document}
%%%%%%%%%%%%%%%%%%%%%%%%%%%%%%%%%%%%%%%%%%%%%%%%%%%%%%%%%%%%%%%

\maketitle

%%%%%%%%%%%%%%%%%%%%%%%%%%%%%%%%%%%%%%%%%%%%%%%%%%%%%%%%%%%%%%%
\begin{abstract}
    We investigate quantum authentication schemes constructed from quantum error-correcting codes. We show that if the code has a property called \emph{purity testing}, then the resulting authentication scheme guarantees the integrity of ciphertexts, not just plaintexts. On top of that, if the code is \emph{strong} purity testing, the authentication scheme also allows the encryption key to be recycled, partially even if the authentication rejects. Such a strong notion of authentication is useful in a setting where multiple ciphertexts can be present simultaneously, such as in interactive or delegated quantum computation. With these settings in mind, we give an explicit code (based on the trap code) that is strong purity testing but, contrary to other known strong-purity-testing codes, allows for natural computation on ciphertexts.
\end{abstract}
%%%%%%%%%%%%%%%%%%%%%%%%%%%%%%%%%%%%%%%%%%%%%%%%%%%%%%%%%%%%%%%

%%%%%%%%%%%%%%%%%%%%%%%%%%%%%%%%%%%%%%%%%%%%%%%%%%%%%%%%%%%%%%%
\section{Introduction}\label{sec:introduction}
%%%%%%%%%%%%%%%%%%%%%%%%%%%%%%%%%%%%%%%%%%%%%%%%%%%%%%%%%%%%%%%
A central topic in cryptography is authentication: how can we make sure that a message remains unaltered when we send it over an insecure channel? How do we protect a file from being corrupted when it is stored someplace where adversarial parties can potentially access it? And, especially relevant in the current era of cloud computing, how can we let an untrusted third party compute on such authenticated data?

Following extensive research on authentication of classical data, starting with the seminal work by Wegman and Carter~\cite{WC81}, several schemes have been proposed for authenticating quantum states~\cite{BCGST02,ABOE08,BGS13}. One notable such scheme is the trap code~\cite{BGS13}, an encoding scheme that surrounds the data with dummy qubits that function as \emph{traps}, revealing any unauthorized attempts to alter the plaintext data. A client holding the classical encryption key can guide a third party in performing computations directly on the ciphertext by sending input-independent auxiliary quantum states that help bypass the traps, and updating the classical key during the computation. The result is an authenticated output ciphertext.

The trap code distinguishes itself from other quantum authentication schemes in two ways. First, individually-authenticated input qubits can be entangled during the computation, but still be de-authenticated individually. This contrasts for example the Clifford code~\cite{ABOE08}, where de-authentication needs to happen simultaneously on all qubits that were involved in the computation, including any auxiliary ones. Second, the trap code allows for `authenticated measurements': if a third party measures a ciphertext, the client can verify the authenticity of the result from the classical measurement outcomes only. These two qualities make the trap code uniquely suited for quantum computing on authenticated data. It was originally designed for its use in quantum one-time programs~\cite{BGS13}, but has found further applications in zero-knowledge proofs for QMA~\cite{BJSW16}, and in quantum homomorphic encryption with verification~\cite{ADSS17}.

The extraordinary structure of the trap code is simultaneously its weakness: an adversary can learn information about the secret key by altering the ciphertext in a specific way, and observing whether or not the result is accepted by the client. Thus, to ensure security after de-authentication, the key needs to be refreshed before another quantum state is authenticated. This need for a refresh inhibits the usefulness of the trap code, because computation on multiple qubits under the trap code requires these qubits to be authenticated under overlapping secret keys.

In recent years, several works have refined the original definition of quantum authentication by Barnum et al.~\cite{BCGST02}. The trap code is secure under the weakest of these definitions~\cite{DNS12}, where only authenticity of the plaintext is guaranteed. But, as argued, it is not under the stronger `total authentication'~\cite{GYZ17}, where no information about the key is leaked if the client accepts the authentication. As Portmann mentions in his work on authentication with key recycling in the abstract-cryptography framework~\cite{P17}, it is not even clear whether the trap code can be regarded as a scheme with \emph{partial} key leakage, as defined in~\cite{GYZ17}, because of the adaptive way in which it can be attacked. In a different direction, Alagic, Gagliardoni, and Majenz~\cite{AGM17} define a notion of quantum ciphertext authentication ($\QCA$), where also the integrity of the ciphertext is guaranteed, and not just that of the plaintext. Ciphertext authentication is incomparable with total authentication: neither one implies the other. Before the current work, it was unknown whether the trap code authenticates ciphertexts.

Barnum et al.~\cite{BCGST02} built schemes for authentication of quantum data based on quantum error-correcting codes that are \emph{purity testing}, meaning that any bit or phase flip on the message is detected with high probability. Portmann~\cite{P17}, working in the abstract-cryptography framework, showed that if the underlying code satisfies a stronger requirement called \emph{strong purity testing}, the resulting authentication scheme allows for complete key recycling in the accept case, and for partial key recycling in the reject case. The trap code can be seen as a purity-testing error-correcting code, but it is not strong purity testing. This is consistent with the observation that keys in the trap code cannot be recycled.

Quantum plaintext authentication with key recycling has been studied before. Oppenheim and Horodecki~\cite{OH05} showed partial key recycling for schemes based on purity testing codes, under a weaker notion of security. Hayden, Leung, and Mayers~\cite{HLM16} adapted Barnum et al.'s construction to use less key and show its authenticating properties in the universal-composability framework. Fehr and Salvail~\cite{FS17} develop a quantum authentication scheme for classical messages that achieves the same key-recycling rate as Portmann~\cite{P17}, but is not based on quantum error-correction and only requires the client to prepare and measure.

%%%%%%%%%%%%%%%%%%%%%%%%%%%%%
\subsection{Our contributions}\label{sec:contributions}
%%%%%%%%%%%%%%%%%%%%%%%%%%%%%
We investigate the relation between (strong) purity testing and quantum ciphertext authentication ($\QCA$), and give a variation on the trap code with stronger security guarantees. We specify our contributions in more detail below.

\paragraph{Section~\ref{sec:QCA-R}: Definition of quantum ciphertext authentication with key recycling \emph{($\QCAR$)}.} We give a new definition for quantum authentication, $\QCAR$, that provides both ciphertext authentication and key recycling, and is thereby strictly stronger than existing definitions. See Figure~\ref{fig:contributions} for a comparison of different notions of authentication.

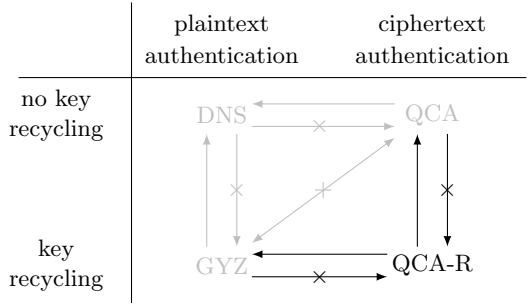
\begin{figure}
    \centering
    \begin{tikzpicture}
    \node at (0.2,0) {\color{lightgray}GYZ};
    \node at (3,0) {QCA-R};
    \node at (0.2,2) {\color{lightgray}DNS};
    \node at (3,2) {\color{lightgray}QCA};
    \draw (-2.5,2.5) -- (4.25,2.5);
    \draw (-1,3.5) -- (-1,-0.5);
    \node at (0.2,3.2) {plaintext};
    \node at (0.2,2.8) {authentication};
    \node at (3,3.2) {ciphertext};
    \node at (3,2.8) {authentication};
    \node at (-2,2.2) {no key};
    \node at (-2,1.8) {recycling};
    \node at (-2,0.2) {key};
    \node at (-2,-0.2) {recycling};
    
    \draw[-latex, lightgray] (0,0.25) -- (0,1.75);
    \draw[-latex, lightgray] (0.4,1.75) -- (0.4,0.25);
    \node at (0.4,1) {\color{lightgray}$\times$};
    
    \draw[-latex] (2.8,0.25) -- (2.8,1.75);
    \draw[-latex] (3.2,1.75) -- (3.2,0.25);
    \node at (3.2,1) {$\times$};
    
    \draw[-latex] (2.4,0.15) -- (0.6,0.15);
    \draw[-latex] (0.6,-0.15) -- (2.4,-0.15);
    \node at (1.5,-0.15) {$\times$};
    
    \draw[-latex, lightgray] (2.5,2.15) -- (0.6,2.15);
    \draw[-latex, lightgray] (0.6,1.85) -- (2.5,1.85);
    \node at (1.5,1.85) {\color{lightgray}$\times$};
    
    \draw[latex-latex, lightgray] (0.6,0.3) -- (2.5,1.7);
    \node at (1.55,1) {\color{lightgray}+};

    \end{tikzpicture}
    \caption{Overview of different definitions of quantum authentication. Three previously defined notions (in gray) and their relations were already known: $\DNS$~\cite{DNS12} is strictly weaker than $\GYZ$~\cite{GYZ17} (total authentication) and $\QCA$~\cite{AGM17}. These last two are incomparable: there exist schemes that satisfy either one, but not the other. On the bottom right, our new definition $\QCAR$ is displayed: it is strictly stronger than both $\GYZ$ and $\QCA$.}
    \label{fig:contributions}
\end{figure}

\paragraph{Section~\ref{sec:SPT-implies-QCA-R}: Purity-testing codes give rise to $\QCA$-secure encryption.} We prove that Barnum et al.'s canonical construction of authentication schemes from purity-testing codes~\cite{BCGST02} produces schemes that are not only plaintext authenticating, but also ciphertext authenticating ($\QCA$). The proof generalizes the proofs in~\cite{BW16} that the trap code and Clifford code are plaintext authenticating, using a different (but still efficient) simulator. Note that our result immediately implies that the trap code is ciphertext authenticating.

\paragraph{Section~\ref{sec:SPT-implies-QCA-R}: Strong-purity-testing codes give rise to $\QCAR$-secure encryption.} Purity-testing codes are generally not sufficient for constructing $\QCAR$ schemes, but strong-purity-testing codes are: we prove that Barnum et al.'s canonical construction achieves $\QCAR$ when a strong-purity-testing code is used as a resource. In case the authenticated message is accepted, the entire key can be reused. Otherwise, all but the quantum-one-time-pad key can be reused.

\paragraph{Section~\ref{sec:strong-trap-code}: A strong-purity-testing version of the trap code.} We give an explicit construction of a strong-purity-testing code that is inspired by the trap code. In this \emph{strong trap code}, the underlying error-correcting code is not only applied to the data qubits, but also to the trap qubits. The result is a quantum authentication scheme which satisfies the strong notion of $\QCAR$, but still maintains the computational properties that make the original trap code such a useful scheme.

\paragraph{Section~\ref{sec:parallel-encryptions}: Security under parallel encryption.} To illustrate the power of recycling key in the reject case, we consider a setting with a different type of key reuse: reusing (part of) a key immediately to authenticate a second qubit, even before the first qubit is verified. We show that, if multiple qubits are simultaneously authenticated using a scheme that is based on a strong-purity-testing code, then de-authenticating some of these qubits does not jeopardize the security of the others, even if their keys overlap. This property is especially important when using the computational capabilities of the strong trap code, since computing on authenticated qubits needs multiple qubits to use overlapping keys.

%%%%%%%%%%%%%%%%%%%%%%%%%%%%%%%%%%%%%%%%%%%%%%%%%%%%%%%%%%%%%%%
\section{Preliminaries}\label{sec:preliminaries}
%%%%%%%%%%%%%%%%%%%%%%%%%%%%%%%%%%%%%%%%%%%%%%%%%%%%%%%%%%%%%%%

%%%%%%%%%%%%%%%%%%%%%%%%%%%%%
\subsection{Notation}\label{sec:general-preliminaries}
%%%%%%%%%%%%%%%%%%%%%%%%%%%%%
We use conventional notation for unitary matrices ($U$ or $V$), pure states ($\ket\psi$ or $\ket\varphi$), and mixed states ($\rho$ or $\sigma$). We reserve the symbol $\tau$ for the completely mixed state $\id/d$, and $\ket{\Phi^+}$ for the EPR pair $\frac{1}{\sqrt{2}}(\ket{00} + \ket{11})$. The $m$-qubit Pauli group is denoted with $\Pauli_m$, and its elements with $P_{\ell}$ where $\ell$ is a $2m$-bit string indicating the bit-flip and phase-flip positions. By convention, $P_0$ is identity.

We often specify the register(s) on which a unitary acts by gray superscripts (as in $U^{\reg{R}}$); it is implicit that the unitary acts as identity on all other registers. The trace norm of a density matrix $\rho$ is written as $\trnorm{\rho}$. The diamond norm of a channel $\Psi$ is written as $\dnorm{\Psi} := \sup_{\rho}\trnorm{(\id \otimes \Psi)(\rho)}$. If we want to talk about the distance between two channels $\Psi$ and $\Psi'$, we use the normalized quantity $\frac{1}{2}\dnorm{ \Psi - \Psi'}$, which we refer to as the \emph{diamond-norm distance}.

%%%%%%%%%%%%%%%%%%%%%%%%%%%%%
\subsection{Quantum authentication}\label{sec:quantum-authentication}
%%%%%%%%%%%%%%%%%%%%%%%%%%%%%
A secret-key quantum encryption scheme consists of three (efficient) algorithms: key generation $\KeyGen$, encryption $\Encrypt_k$, and decryption $\Decrypt_k$. Throughout this work, we will assume that $\KeyGen$ selects a key $k$ uniformly at random from some set $\Key$; our results still hold if the key is selected according to some other distribution. By Lemma B.9 in~\cite{AM17}, we can characterize the encryption and decryption maps as being of the form
\begin{align}
    \Encrypt_k &: \rho^{\reg{M}} \mapsto U_k^{\reg{MT}} (\rho \otimes \sigma_k^{\reg{T}}) (U_k^{\dag})^{\reg{MT}},\label{eq:encrypt}\\
    \Decrypt_k &: \rho^{\reg{MT}} \mapsto \Tr_T \left[(\Pi^{\cacc}_k)^{\reg{T}} \left(U_k^{\dag} \rho U_k^{\reg{MT}} \right) (\Pi^{\cacc}_k)^{\reg{T}} \right] + D_k^{\reg{MT}}\left[(\Pi^{\crej}_k)^{\reg{T}} \left(U_k^{\dag} \rho U_k^{\reg{MT}} \right) (\Pi^{\crej}_k)^{\reg{T}}\right].\label{eq:decrypt}
\end{align}
Here, $M$ is the message register, $\sigma_k$ is some key-dependent \emph{tag} state in register $T$, and $U_k$ is a unitary acting on both. $\Pi^{\cacc}_k$ and $\Pi_k^{\crej}$ are orthogonal projectors onto the support of $\sigma_k$ and its complement, respectively. Finally, $D_k$ is any channel: we will usually assume that $D_k(\cdot) = \Tr_{MT}(\cdot) \otimes \proj\bot^{\reg{M}}$, i.e., it traces out the message and tag register entirely, and replaces the message with some dummy state that signifies a reject. Because of the above characterization, we will often talk about encryption schemes as a keyed collection $\{(U_k, \sigma_k)\}_{k \in \Key}$ of unitaries and tag states.

There are several definitions of the authentication of quantum data. All definitions involve some parameter $\epsilon$; unless otherwise specified, we require $\epsilon$ to be negligibly small in the size of the ciphertext.

The simplest definition is that of plaintext authentication, presented in~\cite{DNS12} (although their definition was in phrased terms of the trace norm), where no guarantees are given about the recyclability of the key.

\begin{definition}[Quantum plaintext authentication ($\DNS$)~\cite{DNS12}]\label{def:DNS} A quantum encryption scheme $\{(U_k, \sigma_k)\}_{k \in \Key}$ is \emph{plaintext $\epsilon$-authenticating} (or $\epsilon$-$\DNS$) if for all CP maps $\advA$ (acting on the message register $M$, tag register $T$, and a side-information register $R$), there exist CP maps $\simacc$ and $\simrej$ such that $ \simS := \simacc + \simrej$ is trace-preserving, and
\begin{align}
    \frac{1}{2}\dnorm*{\expectation_k\left[\Decrypt_k \circ \advA^{\reg{MTR}} \circ \Encrypt_k\right]^{\reg{MR}} - \Bigl(\id^{\reg{M}} \otimes \simacc^{\reg{R}} + \proj{\bot}^{\reg{M}} (\Tr_M \otimes \ \simrej^{\reg{R}}) \Bigr)} \leq \epsilon,\nonumber
\end{align}
where $\Encrypt_k$ and $\Decrypt_k$ are of the form of equations~\eqref{eq:encrypt} and~\eqref{eq:decrypt}.
\end{definition}

The simulator in Definition~\ref{def:DNS} reflects the ideal functionality of an authentication scheme: in the accept case, the message remains untouched, whereas in the reject case, it is completely discarded and replaced with the fixed state $\proj{\bot}$. Any action on the side-information register $R$ is allowed.

\paragraph{The trap code.} \label{page:trap-code} An example of a plaintext-authenticating scheme is the trap code~\cite{BGS13}. This scheme encrypts single-qubit messages by applying a fixed distance-$d$ CSS code $E$ to the message, producing $n$ physical qubits, and then appending $2n$ ``trap" qubits ($n$ computational-basis traps in the state $\proj0$, and $n$ Hadamard-basis traps in the state $\proj+$). The resulting $3n$ qubits are permuted in a random fashion according to a key $k_1$, and one-time padded with a second key $k_2$. At decryption, the one-time pad and permutation are removed, the traps are measured in their respective bases, and the syndrome of the CSS code is checked.\footnote{We differ from the analysis by  Broadbent and Wainewright~\cite{BW16} in that we consider the variant that uses error detection instead of error correction on the data qubits.} The trap code, for a key $k = (k_1,k_2)$, is characterized by $U_k = P_{k_2} \pi_{k_1}(E \otimes \I^{\otimes n} \otimes \H^{\otimes n})$ and $\sigma_k = \proj{0}^{\otimes (3n-1)}$, where $\pi_{k_1}$ is a unitary that permutes the $3n$ qubits. A proof that the trap code is plaintext $(\sfrac{2}{3})^{\sfrac{d}{2}}$-authenticating can be found in e.g.~\cite{BW16}.
\\ \\ \indent Another definition of quantum authentication is presented in~\cite{GYZ17} (where it is called `total authentication'): in this definition, the key should be recyclable in the accept case. This is modeled by revealing the key to the environment after use, and requiring that it is indistinguishable from a completely fresh and uncorrelated key. If that is the case, it can be recycled for another round.

\begin{definition}[Quantum plaintext authentication with key recycling ($\GYZ$)~\cite{GYZ17}]\label{def:GYZ} A quantum encryption scheme $\{(U_k, \sigma_k)\}_{k \in \Key}$ is \emph{plaintext $\epsilon$-authenticating with key recycling} (or $\epsilon$-$\GYZ$) if for all CP maps $\advA$ (acting on the message register $M$, tag register $T$, and a side-information register $R$), there exist CP maps $\simacc$ and $\simrej$ such that $ \simS := \simacc + \simrej$ is trace preserving, and
\begin{align}
    \frac{1}{2}\left\|\expectation_k\left[\rho^{\reg{MR}} \mapsto \Tr_T\left( \Pi^{\cacc}_k U^{\dag}_k \left(\advA^{\reg{MTR}}\left(U_k(\rho \otimes \sigma_k^{\reg{T}})U^{\dag}_k\right)\right) U_k \Pi^{\cacc}_k \right) \otimes \proj{k}\right] - \left(\id^{\reg{M}} \otimes \simacc^{\reg{R}} \otimes \tau_{\Key} \right)\right\|_{\diamond} \leq \epsilon.\nonumber
\end{align}
\end{definition}
Note that Definition~\ref{def:GYZ} only specifies what should happen in the accept case. Nevertheless, it is a strictly stronger definition than $\DNS$ authentication~\cite{AM17}.

The trap code is not plaintext $\epsilon$-authenticating with key recycling for sub-constant $\epsilon$. To see this, consider an adversary $\advA$ that applies $\X$ to (only) the first qubit of the $MT$ register. With probability $\sfrac{2}{3}$, the attack lands on a data qubit or a $\proj{0}$ trap, and is detected. Thus, in the real accept scenario, the key register will contain a mixture of only those keys that permute a $\proj{+}$ into the first position. All other keys are diminished by the projector $\Pi^{\cacc}_k$. Since the ideal scenario contains a mixture of \emph{all} possible keys in the key register, the difference between the two channels is considerable. In practice, if an adversary learns whether the authentication succeeded, she gets information about the positions of the traps.

\paragraph{The Clifford code.} A simple yet powerful code that authenticates plaintexts with key recycling is the Clifford code~\cite{ABOE08}. In this code, we fix a parameter $t$, and set $\sigma_k = \proj{0^t}$ for all $k$, and $U_k$ a uniformly random Clifford on $t+1$ qubits. The Clifford code (and any authentication code that is based on a 2-design) is plaintext $\epsilon$-authenticating with key recycling for $\epsilon = O(2^{-t})$~\cite{AM17}.
\\ \\ \indent Strengthening Definition~\ref{def:DNS} in a different direction, Alagic, Gagliardoni, and Majenz recently introduced the notion of quantum ciphertext authentication~\cite{AGM17}. This notion does not limit the amount of key leaked, but requires that if authentication accepts, the entire \emph{ciphertext} was completely untouched.

\begin{definition}[Quantum ciphertext authentication ($\QCA$)~\cite{AGM17}]\label{def:QCA} A quantum encryption scheme $\{(U_k, \sigma_k = \sum_r p_{k,r} \proj{\varphi_{k,r}})\}_{k \in \Key}$ is \emph{ciphertext $\epsilon$-authenticating} (or $\epsilon$-$\QCA$) if it is plaintext $\epsilon$-authenticating as in Definition~\ref{def:DNS}, and the accepting simulator $\simacc$ is of the form
\begin{align}
    \simacc: \rho^{\reg{R}} \mapsto \expectation_{k',r} \left[ \bra{\varphi_{k',r}}^{\reg{T}}\bra{\Phi^+}^{\reg{M_1M_2}} U_{k'}^{\dag} \advA^{\reg{M_1TR}}\left(U_{k'}^{\reg{M_1T}}\rho_{k',r}^{\reg{RM_1M_2T}}U_{k'}^{\dag}\right)U_{k'} \ket{\varphi_{k',r}}\ket{\Phi^+}\right].\nonumber
\end{align}
where $\rho_{k',r} := \rho^{\reg{R}} \otimes \proj{\Phi^+}^{\reg{M_1M_2}} \otimes \proj{\varphi_{k',r}}^{\reg{T}}$ is the input state before (simulated) encryption.
\end{definition}
In $\QCA$, the accepting simulator tests whether the message remains completely untouched by encrypting half of an EPR pair (stored in register $M_1$) as a `dummy message', under a key $k'$ that it generates itself. It remembers the randomness $r$ used in creating the tag state $\sigma_k$, so that it can test very accurately whether the tag state was untouched. Because $\simacc$ remembers the randomness, a scheme that appends a qubit at the end of its ciphertexts, but never checks its state at decryption time, cannot be ciphertext authenticating. The Clifford code \emph{is} $\QCA$~\cite{AGM17}, as is the trap code (see Section~\ref{sec:SPT-implies-QCA-R}).

In general, key recycling as in Definition~\ref{def:GYZ} does not imply $\QCA$. To see this, take any scheme $\{(U_k, \sigma_k)\}_{k \in \Key}$ that is plaintext authenticating with key recycling, and alter it by appending a qubit in the fully mixed state to $\sigma_k$ (and extending $U_k$ to act as identity on this qubit). This scheme still satisfies Definition~\ref{def:GYZ}, but cannot be ciphertext authenticating, because attacks on this last qubit are not noticed in the real scenario. Conversely, not all ciphertext-authenticating schemes have key recycling. Take any scheme that is $\QCA$, and alter it by adding one extra bit $b$ of key, and setting $\sigma_{kb} := \sigma_{k} \otimes \proj{b}$ and $U_{kb} := U_k \otimes \I$, effectively appending the bit of key at the end of the ciphertext. This scheme still satisfies Definition~\ref{def:QCA}, but leaks at least one bit of key.\footnote{We thank Gorjan Alagic and Christian Majenz for providing these example schemes that show the separation between Definitions~\ref{def:GYZ} and~\ref{def:QCA}.} For an overview of the relations between $\DNS$, $\GYZ$, and $\QCA$, refer to Figure~\ref{fig:contributions} on page~\pageref{fig:contributions}.

%%%%%%%%%%%%%%%%%%%%%%%%%%%%%
\subsection{(Strong) purity testing in quantum error correction}\label{sec:purity-testing}
%%%%%%%%%%%%%%%%%%%%%%%%%%%%%
An $[[n,m]]$ quantum error-correcting code (QECC), characterized by a unitary operator $V$, encodes a message $\rho$ consisting of $m$ qubits into a codeword $V(\rho \otimes \proj{0^t})V^{\dag}$ consisting of $n$ qubits, by appending $t := n - m$ tags $\proj0$, and applying the unitary $V$. Decoding happens by undoing the unitary $V$, and measuring the tag register in the computational basis. The measurement outcome is called the syndrome: an all-zero syndrome indicates that no error-correction is necessary. In this work, we will only use the error-detection property of QECCs, and will not worry about how to correct the message if a non-zero syndrome is measured. If that happens, we will simply discard the message (i.e., reject).

For any bit string $x \in \{0,1\}^m$, let $\ket{x_L}$ (for ``logical $\ket{x}$") denote a valid encoding of $\ket{x}$, i.e., a state that will decode to $\ket{x}$ without error. A defining feature of any QECC is its distance: the amount of bit and/or phase flips required to turn one valid codeword into another. If we want to be explicit about the distance $d$ of an $[[n,m]]$ code, we will refer to it as an $[[n,m,d]]$ code.

\begin{definition}[Distance]\label{def:distance-QECC}
	The \emph{distance} of an $[[n,m]]$ code is the minimum weight of a Pauli $P$ such that $P\ket{x_L} = \ket{y_L}$ for some $x \neq y$, with $x,y \in \{0,1\}^m$.
\end{definition}

In a cryptographic setting, it can be useful to select a code from a set of codes $\{V_k\}_{k \in \mathcal{K}}$ for some key set $\mathcal{K}$. We will again assume that the key $k$ is selected uniformly at random. 

Following~\cite{BCGST02} and~\cite{P17}, we restrict our attention to codes for which applying a Pauli to a codeword is equivalent to applying a (possibly different) Pauli directly to the message and tag register. In other words, the unitary $V$ must be such that for any  $P_\ell \in \Pauli_{m+t}$, there exists a $P_{\ell'} \in \Pauli_{m+t}$ and a $\theta \in \mathbb{R}$ such that $P_{\ell}V = e^{i\theta}VP_{\ell'}$. With our attention restricted to codes with this property, we can meaningfully define the following property:

\begin{definition}[Purity testing~\cite{BCGST02}]\label{def:purity-testing}
A set of codes $\{V_k\}_{k \in \mathcal{K}}$ is \emph{purity testing} with error $\varepsilon$ if for any Pauli $P_{\ell} \in \Pauli_{m+t} \backslash \{\I^{\otimes (m+t)}\}$,
\begin{align}
    \Pr_k\left[V^{\dag}_k P_{\ell} V_k \in (\Pauli_m \backslash \{\I^{\otimes m}\}) \otimes \{\I,\Z\}^{\otimes t}\right] \leq \varepsilon.\nonumber
\end{align}
\end{definition}
In words, for any non-identity Pauli, the probability (over the key) that the Pauli alters the message but is not detected (i.e., no tag bit is flipped) is upper bounded by $\epsilon$.

The trap code (see page~\pageref{page:trap-code}) based on an $[[n,1,d]]$ CSS code, without the final quantum one-time pad, is a purity-testing code with error $(\sfrac{2}{3})^{\sfrac{d}{2}}$~\cite{BGS13}. In our framework, the trap code is described as a QECC with $m = 1$, $t = 3n-1$, and $V_k = \pi_k(E \otimes \I^{\otimes n} \otimes \H^{\otimes n})$.

Note that purity-testing codes do not necessarily detect \emph{all} Pauli attacks with high probability: it may well be that a Pauli attack remains undetected, because it acts as identity on the message. Flipping the first bit of a trap-code ciphertext is an example of such an attack: it remains undetected with probability $\sfrac{1}{3}$ (if it hits a $\ket+$ trap), but unless it is detected, it also does not alter the message. An attacker may use this fact to learn information about the permutation $\pi_k$ by observing whether or not the QECC detects an error.

The above exploitation of purity-testing codes has led Portmann to consider a stronger notion of purity testing that should allow for keys to be safely reusable. In this definition, even the Paulis that act as identity on the message should be detected:

\begin{definition}[Strong purity testing~\cite{P17}]\label{def:strong-purity-testing}
A set of codes $\{V_k\}_{k \in \mathcal{K}}$ is \emph{strong purity testing} with error $\varepsilon$ if for any Pauli $P_{\ell} \in \Pauli_{m+t} \backslash \{\I^{\otimes (m+t)}\}$,
\begin{align}
    \Pr_k\left[V^{\dag}_k P_{\ell} V_k \in \Pauli_m \otimes \{\I,\Z\}^{\otimes t}\right] \leq \varepsilon.\nonumber
\end{align}
\end{definition}

The Clifford code is strong purity testing with error $2^{-t}$, as is any other unitary 2-design~\cite{P17}. As informally discussed above, the trap code is not strong purity testing for any small $\epsilon$.

Barnum et al.~\cite{BCGST02} described a canonical method of turning a QECC set $\{V_{k_1}\}_{k_1 \in \mathcal{K}_1}$ into a symmetric-key encryption scheme. The encryption key $k$ consists of two parts: the key $k_1 \in \mathcal{K}_1$ for the QECC, and an additional one-time pad key $k_2 \in \{0,1\}^{2(m+t)}$. The encryption map is then defined by setting $U_{k_1,k_2} := P_{k_2}V_{k_1}$, and $\sigma_{k_1,k_2} = \proj{0^t}$. Since $\sigma_{k_1,k_2}$ is key-independent, the projectors $\Pi^{\cacc} = \proj{0^t}$ and $\Pi^{\crej} = \id - \proj{0^t}$ are key-independent as well. In Construction~\ref{con:qecc-to-skqes}, the complete protocol is described. In~\cite{BGS13}, protocols of this form are called ``encode-encrypt schemes".

\begin{construction}
\caption{Barnum et al.'s canonical construction~\cite{BCGST02} of a symmetric-key encryption scheme from an $[[m+t,m]]$ quantum error-correcting code $\{V_{k_1}\}_{k_1 \in \mathcal{K}_1}$.}\label{con:qecc-to-skqes}
    \begin{algorithmic}
    \State \textbf{Key generation:} Sample $k_1 \gets \mathcal{K}_1$. Sample $k_2 \gets \mathcal{K}_2 = \{0,1\}^{2(m+t)}$.
    \State \textbf{Encryption:} $\rho^{\reg{M}} \mapsto P_{k_2}^{\reg{MT}}V_{k_1}^{\reg{MT}}(\rho^M \otimes \proj{0^t}^{\reg{T}})V_{k_1}^{\reg{MT}}P_{k_2}^{\reg{MT}}$.
    \State \textbf{Decryption:} $\rho^{\reg{MT}} \mapsto \bra{0^t} \left(V_{k_1}^{\dag}P_{k_2}^{\dag} \rho P_{k_2}V_{k_1} \right) \ket{0^t} \ \ + \ \  \proj{\bot}^{\reg{M}} \otimes \Tr_{M}\left[\sum\limits_{i \neq 0^t} \bra{i} \left(V_{k_1}^{\dag}P_{k_2}^{\dag} \rho P_{k_2}V_{k_1} \right) \ket{i}\right]$
    \end{algorithmic}
\end{construction}

When using Construction~\ref{con:qecc-to-skqes} with a strong-purity-testing code, plaintext authentication with key recycling is achieved, even with partial key recycling in the reject case~\cite{P17}. If just a purity-testing code is used for the construction, the resulting encryption scheme is plaintext authenticating~\cite{BCGST02}, but not necessarily with key recycling (the trap code is a counterexample).

%%%%%%%%%%%%%%%%%%%%%%%%%%%%%%%%%%%%%%%%%%%%%%%%%%%%%%%%%%%%%%%
\section{Quantum ciphertext authentication with key recycling ($\QCAR$)}\label{sec:QCA-R}
%%%%%%%%%%%%%%%%%%%%%%%%%%%%%%%%%%%%%%%%%%%%%%%%%%%%%%%%%%%%%%%
In this section, we will define a notion of quantum authentication that is stronger than all of Definitions~\ref{def:DNS},~\ref{def:GYZ}, and~\ref{def:QCA}. We will show that Construction~\ref{con:qecc-to-skqes}, when used with a strong-purity-testing code, results in an authentication scheme in this new, stronger sense.

\begin{definition}[Quantum ciphertext authentication with key recycling ($\QCAR$)]\label{def:QCAR} A quantum encryption scheme $\{(U_k, \sigma_k = \sum_r p_{k,r} \proj{\varphi_{k,r}})\}_{k \in \Key}$ is \emph{ciphertext $\epsilon$-authenticating with key recycling} (or $\epsilon$-$\QCAR$), with key recycling function $f$, if for all CP maps $\advA$ (acting on the message register $M$, tag register $T$, and a side-information register $R$), there exists a CP map $\simrej$ such that
\begin{align}
\mathfrak{R} : \rho^{\reg{MR}} \mapsto \mathbb{E}_k \Big[\Tr_T &\left(\Pi^{acc} \left(U_k^{\dag} \mathcal{A}^{\reg{MTR}}\left(U_k^{\reg{MT}} (\rho \otimes \sigma_k^{\reg{T}}) U_k^{\dag}\right) U_k \right) \Pi^{acc} \right) \otimes \proj{k}\nonumber \\
\ \  + \ \proj{\bot}^{\reg{M}} \otimes \ \Tr_{MT} &\left(\Pi^{rej} \left(U_k^{\dag} \mathcal{A}^{\reg{MTR}}\left(U_k^{\reg{MT}} (\rho \otimes \sigma_k^{\reg{T}}) U_k^{\dag}\right) U_k \right) \Pi^{rej} \right) \otimes \proj{f(k)}\Big]\nonumber
\end{align}
is $\epsilon$-close in diamond-norm distance to the ideal channel,
\begin{align}
\mathfrak{I} : \rho^{\reg{MR}} \mapsto \left(\id^{\reg{M}} \otimes \mathcal{S}^{\cacc}\right)(\rho^{\reg{MR}}) \otimes \tau_{\Key} \ \ \  + \ \ \ \proj{\bot}^{\reg{M}} \otimes \ \mathcal{S}^{\crej}(\rho^{\reg{R}}) \otimes \mathbb{E}_k \left[\proj{f(k)}\right],\nonumber
\end{align}
where $\simS := \simacc + \simrej$ is trace preserving, and $\simacc$ is as in Definition~\ref{def:QCA} of $\QCA$, that is,
\begin{align}
    \simacc: \rho^{\reg{R}} \mapsto \expectation_{k',r} \left[ \bra{\varphi_{k',r}}^{\reg{T}}\bra{\Phi^+}^{\reg{M_1M_2}} U_{k'}^{\dag} \advA^{\reg{M_1TR}}\left(U_{k'}^{\reg{M_1T}}\rho_{k',r}^{\reg{RM_1M_2T}}U_{k'}^{\dag}\right)U_{k'} \ket{\varphi_{k',r}}\ket{\Phi^+}\right]\nonumber
\end{align}
for $\rho_{k',r} := \rho^{\reg{R}} \otimes \proj{\Phi^+}^{\reg{M_1M_2}} \otimes \proj{\varphi_{k',r}}^{\reg{T}}$.
\end{definition}

The first condition (closeness of the real and ideal channel) is a strengthening of Definition~\ref{def:GYZ}: following Portmann~\cite{P17}, we also consider which part of the key can be recycled in the reject case. If the recycling function $f$ is the identity function, all of the key can be recycled. If $f$ maps all keys to the empty string, then no constraints are put on key leakage in the reject case.

$\QCAR$ strengthens both $\GYZ$ and $\QCA$, but not vice versa: the schemes from Section~\ref{sec:quantum-authentication} that separate the two older notions are immediately examples of schemes that are $\GYZ$ or $\QCA$ but cannot be $\QCAR$.

%%%%%%%%%%%%%%%%%%%%%%%%%%%%%
\subsection{Constructing QCA-R from any strong-purity-testing code}\label{sec:SPT-implies-QCA-R}
%%%%%%%%%%%%%%%%%%%%%%%%%%%%%
It was already observed that if a set of quantum error-correcting codes $\{V_{k_1}\}_{k_1 \in \mathcal{K}_1}$ is purity testing, then the encryption scheme resulting from Construction~\ref{con:qecc-to-skqes} is plaintext authenticating~\cite{BCGST02}. We strengthen this result by showing that the construction turns purity-testing codes into \emph{ciphertext}-authenticating schemes (Theorem~\ref{thm:PT-implies-QCA}), and strong-purity-testing codes into $\QCAR$ schemes (Theorem~\ref{thm:SPT-implies-QCA-R}). Only purity testing is in general not enough to achieve $\QCAR$: the trap code is again a counterexample.

\begin{theorem}\label{thm:PT-implies-QCA}
Let $\{V_{k_1}\}_{k_1 \in \mathcal{K}_1}$ be a purity-testing code with error $\epsilon$. The encryption scheme resulting from Construction~\ref{con:qecc-to-skqes} is quantum ciphertext $\epsilon$-authenticating ($\epsilon$-$\QCA$).
\end{theorem}

\begin{proof}[Sketch]
    In order to prove Theorem~\ref{thm:PT-implies-QCA}, we define a simulator that runs the adversary on encrypted halves of EPR pairs, so that the simulator is of the correct form for Definition~\ref{def:QCA}. We prove that the ideal and the real channel are close by considering the accept and the reject cases separately, and by showing that they are both $\sfrac{\epsilon}{2}$-close. First, we decompose the adversarial attack into Paulis by Pauli twirling~\cite{DCEL09} it with the quantum-one-time-pad encryption from Construction~\ref{con:qecc-to-skqes}. In the accept case, the difference between the real and the ideal scenario lies in those attacks that are accepted in the real case, but not in the ideal case. These are exactly those Paulis that, after conjugation with the key $k_1$ that indexes the purity-testing code, are in the set $(\Pauli_m \otimes \{\I,\Z\}^{\otimes t}) \backslash (\{\I^{\otimes m}\} \otimes \{\I,\Z\}^{\otimes t}) = (\Pauli_m\backslash\{\I^{\otimes m}\}) \otimes \{\I,\Z\}^{\otimes t}$. The purity-testing property guarantees that the probability over $k_1$ of a Pauli attack landing in this set is small. The reject case is similar.
    \qed
\end{proof}

A full proof of Theorem~\ref{thm:PT-implies-QCA} is in Appendix~\ref{appendix:proof-of-PT-implies-QCA}. The proof of Theorem~\ref{thm:SPT-implies-QCA-R} below uses the same techniques. It follows the proof structure of~\cite[Theorem~3.5]{P17}, but with a simulator that is suitable for $\QCAR$.  

\begin{theorem}\label{thm:SPT-implies-QCA-R}
Let $\{V_{k_1}\}_{k_1 \in \mathcal{K}_1}$ be a strong-purity-testing code with error $\varepsilon$. The encryption scheme resulting from Construction~\ref{con:qecc-to-skqes} is quantum ciphertext $(\sqrt{\epsilon}+\frac{3}{2}\epsilon)$-authenticating with key recycling \emph{($\epsilon$-$\QCAR$)}, with recycling function $f(k_1,k_2) := k_1$.
\end{theorem}

\begin{proof}
    Let $\advA$ be an adversary as in Definition~\ref{def:QCAR}. Define a simulator $\simS$ on the side-information register $R$ as follows: prepare an EPR pair $\proj{\Phi^+}$ in the register $M_1M_2$ and encrypt the first qubit using a freshly sampled key $(k_1',k_2') \in \Key := \Key_1 \times \Key_2$ (that is, initialize the tag register $T$ in the state $\proj{0^t}$, and apply $P_{k_2'}V_{k_1'}$ to $M_1T$). Then, run the adversary on the registers $M_1TR$, keeping $M_2$ to the side. Afterwards, run the decryption procedure by undoing the encryption unitary and measuring whether the registers $M_1M_2T$ are still in the state $\proj{\Phi^+,0^t} \ (= \proj{\Phi^+} \otimes \proj{0^t})$. If so, accept, and if not, reject. Note that this simulator is of the required form in the accept case (see Definition~\ref{def:QCAR}).
    
    We show that for this simulator, the distance $\frac{1}{2}\|\mathfrak{I} - \mathfrak{R}\|_{\diamond}$ between the ideal and the real channel is upper bounded by $\sqrt{\varepsilon} + \frac{3}{2}\epsilon$. Let $\rho^{\reg{MRE}}$ be any quantum state on the message register, side-information register, and an environment register $E$. Let $U^{\reg{MTR}}$ be a unitary\footnote{We can assume unitarity without loss of generality: if the adversary's actions are not unitary, we can dilate the channel into a unitary one by adding another environment and tracing it out afterwards. In the proof, the environment takes on the same role as the side-information register $R$, so we omit it for simplicity.} map representing the adversarial channel $\advA$, and let $\Real_{k_1,k_2}$ and $\Ideal_{k_1,k_2}$ be the effective output states in the real and ideal world, respectively:
    \begin{align}
        \Real_{k_1,k_2} &:= V^{\dag}_{k_1} P^{\dag}_{k_2}U^{\reg{MTR}}P_{k_2}^{\reg{MT}}V_{k_1}^{\reg{MT}} (\rho \otimes \proj{0^t}) V^{\dag}_{k_1} P^{\dag}_{k_2}U^{\dag}P_{k_2}V_{k_1},\\
    \Ideal_{k_1,k_2} &:= V^{\dag}_{k_1} P^{\dag}_{k_2}U^{\reg{M_1TR}}P_{k_2}^{\reg{M_1T}}V_{k_1}^{\reg{M_1T}}
        (\rho \otimes \proj{0^t, \Phi^+})
        V^{\dag}_{k_1} P^{\dag}_{k_2}U^{\dag}P_{k_2}V_{k_1}.
    \end{align}
    Then we can write the result of the real and the ideal channels as
    \begin{align}
        \mathfrak{R}(\rho) &= \expectation_{k_1,k_2} \left[\bra{0^t}^{\reg{T}} \Real_{k_1,k_2}\ket{0^t} \otimes \proj{k_1k_2} \ + \  \proj{\bot}^{\reg{M}} \otimes \ \Tr_{M}\left(\sum_{i\neq 0^t} \bra{i}^{\reg{T}} \Real_{k_1,k_2}\ket{i}\right) \otimes \proj{k_1}\right],\\
        \mathfrak{I}(\rho) &= \expectation_{k_1',k_2'} \left[\bra{\Phi^+,0^t}^{\reg{M_1M_2T}} \Ideal_{k_1',k_2'}\ket{\Phi^+0^t} \otimes \tau_{\Key} \ + \  \proj{\bot}^{\reg{M}} \otimes \ \Tr_{M}\left(\sum_{i\neq (\Phi^+,0^t)} \bra{i}^{\reg{M_1M_2T}} \Ideal_{k_1',k_2'}\ket{i}\right) \otimes \tau_{\Key_1}\right].
    \end{align}
    These expressions are obtained simply by plugging in the description of the authentication scheme (see Construction~\ref{con:qecc-to-skqes}) and the simulator into the channels of Definition~\ref{def:QCAR}. Since the accept states are orthogonal to the reject states in the $M$ register, and since the key states are all mutually orthogonal, the distance $\frac{1}{2}\|\mathfrak{I}(\rho) - \mathfrak{R}(\rho)\|_{\tr}$ can be written as
    \begin{align}
        &\expectation_{k_1,k_2} \frac{1}{2} \left\|\expectation_{k_1',k_2'} \left(\bra{\Phi^+,0^t}\Ideal_{k_1',k_2'}\ket{\Phi^+,0^t}\right) - \bra{0^t}\Real_{k_1,k_2}\ket{0^t}\right\|_{\tr} \nonumber\\
        &\quad + \ \ \expectation_{k_1} \frac{1}{2} \left\| \expectation_{k_1',k_2'} \left(\Tr_M \sum_{i \neq (0^t,\Phi^+)} \bra{i}\Ideal_{k_1',k_2'}\ket{i}\right) - \expectation_{k_2} \left(\Tr_M \sum_{i \neq 0^t} \bra{i} \Real_{k_1,k_2} \ket{i}\right)\right\|_{\tr} \ .\label{eq:distance-to-bound}
    \end{align}
    
    For a full derivation, see Appendix~\ref{appendix:proof-QCAR-derivation-distance}. We can thus focus on bounding the two terms in equation~\eqref{eq:distance-to-bound}, for accept and reject, separately. Note the difference between the two terms: in the reject case, the expectation over the one-time pad key $k_2$ does not have to be brought outside of the trace norm, since it is not recycled after a reject. This will make bounding the second term in equation~\eqref{eq:distance-to-bound} the simpler of the two, so we will start with that one.
    
    Decompose the attack as $U^{\reg{MTR}} = \sum_{\ell} \alpha_{\ell} P_{\ell}^{\reg{MT}} \otimes U_{\ell}^{\reg{R}}$. Intuitively, the two states inside the second trace norm differ on those Paulis $P_{\ell}$ that are rejected in the ideal scenario, but not in the real one. The strong-purity-testing property promises that these Paulis are very few. However, we have to be careful, because the simulator independently generates its own set of keys. We will now bound the second term in equation~\eqref{eq:distance-to-bound} more formally. 
    
    By rearranging sums, commuting Paulis, and applying projectors (for details: see Appendix~\ref{appendix:proof-QCAR-derivation-real-reject}), we can rewrite the second term inside the trace norm, the state in the real reject case for $k_1$, as
    \begin{align}
        \expectation_{k_2} \left(\Tr_M \sum_{i \neq 0^t} \bra{i} \Real_{k_1,k_2}\ket{i}\right) = \Tr_M \left( \sum_{\ell \ : \ V_{k_1}^{\dag} P_{\ell} V_{k_1} \not\in \Pauli_{\real}} \abs{\alpha_{\ell}}^2 U_{\ell}^{\reg{R}} \rho^{\reg{MR}} U_{\ell}^{\dag}\right),\label{eq:real-reject-rewrite}
    \end{align}
    where $\Pauli_{\real}$ contains the Paulis that are accepted by the real projector, i.e., $\Pauli_{\real} := \Pauli_m \otimes \{\I,\Z\}^{\otimes t}$. Similarly, defining $\Pauli_{\ideal} := \{\I^{\otimes m}\} \otimes \{\I,\Z\}^{\otimes t}$ to be the set of Paulis that are allowed by the ideal projector, the resulting state in the reject case is
    
    \begin{align}
        \expectation_{k_1',k_2'} \left(\Tr_M \sum_{i \neq (0^t,\Phi^+)} \bra{i}\Ideal_{k_1',k_2'}\ket{i}\right) &= \Tr_M \left( \sum_{\ell \neq 0} \expectation_{\substack{k_1' \in \Key_1\\V_{k_1'}^{\dag} P_{\ell} V_{k_1'} \not\in \Pauli_{\ideal}}} \abs{\alpha_{\ell}}^2 U_{\ell}^{\reg{R}} \rho^{\reg{MR}} U_{\ell}^{\dag}\right)\label{eq:ideal-reject}\\
        &\approx_{\epsilon} \Tr_M \left( \sum_{\ell \neq 0} \expectation_{k_1' \in \Key_1} \abs{\alpha_{\ell}}^2 U_{\ell}^{\reg{R}} \rho^{\reg{MR}} U_{\ell}^{\dag}\right)\label{eq:ideal-reject-close} \ ,
    \end{align}
    where the approximation sign means that the trace distance between the two states is upper bounded by $\epsilon$. The closeness follows from the strong-purity-testing property of the code: the two states differ in those keys $k_1'$ for which $V_{k_1'}^{\dag}P_{\ell}V_{k_1'} \in \Pauli_{\ideal} \subseteq \Pauli_{\real}$, and for any non-identity Pauli $P_{\ell}$, this set is small by strong purity testing. Combined with the facts that $\tr(U_{\ell}\rho U_{\ell}^{\dag}) = 1$ and $\sum_{\ell}\abs{\alpha_{\ell}}^2 = 1$, it follows that the states in equations~\eqref{eq:ideal-reject} and~\eqref{eq:ideal-reject-close} are $\epsilon$-close. Note that none of the terms in equation~\eqref{eq:ideal-reject-close} depends on $k_1'$, so we can remove the expectation over it.
    
    Applying the triangle inequality (twice), the second term in equation~\eqref{eq:distance-to-bound} is found to be small:
    \begin{align}
        &\phantom{\leq\ }\expectation_{k_1} \frac{1}{2} \left\| \expectation_{k_1',k_2'} \left(\Tr_M \sum_{i \neq (0^t,\Phi^+)} \bra{i}\Ideal_{k_1',k_2'}\ket{i}\right) - \expectation_{k_2} \left(\Tr_M \sum_{i \neq 0^t} \bra{i} \Real_{k_1,k_2} \ket{i}\right)\right\|_{\tr}\\
        &\leq
        \frac{\epsilon}{2} + \expectation_{k_1} \frac{1}{2} \left\|\Tr_M\left( \sum_{\ell \ : \ V_{k_1}^{\dag} P_{\ell} V_{k_1} \in \Pauli_{\real} \backslash \{\I^{\otimes (m + t)}\}} \abs{\alpha_{\ell}}^2 U_{\ell} \rho U_{\ell}^{\dag}\right)\right\|_{\tr}\\
        &\leq
        \frac{\epsilon}{2} + \frac{1}{2} \expectation_{k_1} \sum_{\ell \ : \ V_{k_1}^{\dag} P_{\ell} V_{k_1} \in \Pauli_{\real} \backslash \{\I^{\otimes (m + t)}\}} \abs{\alpha_{\ell}}^2,
    \end{align}
    which we can upper bound by $\epsilon$ by applying the strong-purity-testing property once more.
    
    Next, we bound the first term of equation~\eqref{eq:distance-to-bound}: the difference between the ideal and the real channel in the accept case. The strategy is identical to the reject case that we just treated, but because we want to recycle both $k_1$ and $k_2$ in the accept case, we have to be more careful. The state in the real scenario, $\bra{0^t} \Real_{k_1,k_2} \ket{0^t}$, cannot be rewritten into the compact form of, e.g., equation~\eqref{eq:real-reject-rewrite}, because we cannot average over the Pauli key $k_2$. Using a technical lemma from~\cite{P17} and Jensen's inequality in order to take the expectation over the keys inside, we obtain the bound
    \begin{align}
        \expectation_{k_1,k_2} \left\|\expectation_{k_1',k_2'} \left(\bra{\Phi^+,0^t}\Ideal_{k_1',k_2'}\ket{\Phi^+,0^t}\right) - \bra{0^t}\Real_{k_1,k_2}\ket{0^t}\right\|_{\tr} \leq \frac{\epsilon}{2} + \sqrt{\epsilon}.\label{eq:QCAR-bound-accept-case}
    \end{align}
    For a full derivation, see Appendix~\ref{appendix:proof-QCAR-accept}.
    
    We have now upper bounded $\frac{1}{2}\| \mathfrak{I}(\rho) - \mathfrak{R}(\rho)\|_{\tr} \  \leq \sqrt{\epsilon} + \frac{3}{2}\epsilon$ for any state $\rho^{\reg{MRE}}$, resulting in $\frac{1}{2}\|\mathfrak{I} - \mathfrak{R}\|_{\diamond} \leq \sqrt{\epsilon} + \frac{3}{2}\epsilon$, as desired.
    \qed 
\end{proof}

%%%%%%%%%%%%%%%%%%%%%%%%%%%%%%%%%%%%%%%%%%%%%%%%%%%%%%%%%%%%%%%
\section{A strong-purity-testing variation on the trap code}\label{sec:strong-trap-code}
%%%%%%%%%%%%%%%%%%%%%%%%%%%%%%%%%%%%%%%%%%%%%%%%%%%%%%%%%%%%%%%

Theorem~\ref{thm:SPT-implies-QCA-R} already gives us a quantum-ciphertext-authenticating code with key recycling: the Clifford code. However, the Clifford code is not very well suited for quantum computing on authenticated data. In this section, we therefore present a strong-purity-testing variation on the trap code, the \emph{strong trap code}, which does allow for computation on the ciphertexts in a meaningful and efficient way. By Theorem~\ref{thm:SPT-implies-QCA-R}, this construction immediately gives rise to a ciphertext authentication scheme with key recycling ($\QCAR$). Note that the strong trap code is also secure in Portmann's abstract-cryptography definition of quantum plaintext authentication with key recycling~\cite{P17}.

\subsection{Benign distance and weight sparsity}
The strong trap code requires the existence of a family of quantum error-correcting codes with two specific properties: a high benign distance, and weight sparsity. We specify these properties here.

If a QECC has distance $d$, it is not necessarily able to detect all Pauli errors of weight less than $d$. For example, if one of the qubits in a codeword is in the state $\ket0$, then a Pauli-$\Z$ remains undetected. In general, any Pauli error that stabilizes all codewords will remain undetected by the code. Of course, such an error does not directly cause harm or adds noise to the state, because it effectively performs the identity operation. However, in an adversarial setting, even such `benign' Pauli errors indicate that someone tried to modify the state.

We consider an alternative distance measure for quantum error-correcting codes that describes the lowest possible weight of a stabilizer:
	
\begin{definition}[Benign distance]
	The \emph{benign distance} of an $[[n,m]]$ code is the minimum weight of a non-identity Pauli $P_{\ell}$ such that $P_{\ell}\ket{x_L} = \ket{x_L}$ for all $x \in \{0,1\}^m$. If such $P_{\ell}$ does not exist, the benign distance is $\infty$.
\end{definition}
To distinguish the benign distance from the notion of difference defined in Definition~\ref{def:distance-QECC}, we will often use the term \emph{conventional distance} to refer to the latter.

The benign distance in a fixed relation to the conventional distance. For example, the $[[7,4]]$ Steane code has distance 3, but benign distance 4. On the other hand, the $[[49,1]]$ concatenated Steane code has distance 9, but a benign distance of only 4 (any non-identity stabilizer for the $[[7,4]]$ Steane code is also a stabilizer on the $[[49,1]]$ code if it is concatenated with identity on the other blocks). Even though the two quantities do not bound each other in general, we observe that the benign distance of weakly self-dual \emph{CSS codes} (i.e., CSS codes constructed from a weakly self-dual classical code) grows with their conventional distance. See Lemma~\ref{lem:css-benign-distance} in Appendix~\ref{appendix:high-benign-distance}. 

We define a second property of interest: \emph{weight sparsity}. Intuitively, weight sparsity means that for any set of $\X$-, $\Y$-, and $\Z$-weights, randomly selecting a Pauli operator with those weights only yields a stabilizer with very small probability. This probability should shrink whenever the codeword length grows; for this reason, we consider weight sparsity as a property of code \emph{families} rather than of individual codes.
	
\begin{definition}[Weight-sparse code family]\label{def:weight-sparse}
	Let $(E_i)_{i\in\mathbb{N}}$ be a family of quantum error-correcting codes with parameters $[[n(i), m(i), d(i)]]$. For each $i \in \mathbb{N}$, and for all non-negative integers $x,y,z$ such that $x + y + z \leq n(i)$, let $A_i(x,y,z)$ denote the set of $n(i)$-qubit Paulis with $\X$-weight $x$, $\Y$-weight $y$, and $\Z$-weight $z$. Let $B_i(x,y,z)$ denote set of benign Paulis in $A_i(x,y,z)$.
		
	The family $(E_i)_{i \in \mathbb{N}}$ is \emph{weight-sparse} if the function
	\[
	f(i) := \max_{x + y + z \leq n(i)} \frac{| B_i(x,y,z) |}{ | A_i(x,y,z) | }
	\]
	is negligible\footnote{A function $f(x)$ is negligible in $x$ if for all $c \in \mathbb{N}$, there exists an $x_0$ such that for all $x \geq x_0$, $f(x) < x^{-c}$. This definition is extended by stating that a function $f(x)$ is negligible in another function $g(x)$ if for all $c \in \mathbb{N}$, there exists an $x_0$ such that for all $x \geq x_0$, $f(x) < (g(x))^{-c}$.} in $n(i)$.
\end{definition}

In Appendix~\ref{appendix:construction-of-code-family}, we construct a weight-sparse family of weakly self-dual CSS codes that have benign distance $O(\sqrt{n(i)})$, where $n(i)$ is the codeword length of the $i$th code in the family. The CSS codes are constructed from a punctured version of classical Reed--Muller codes~\cite{Pre97}.

\subsection{The strong trap code}
We present a modified version of the trap code, which we call the \emph{strong trap code}. Contrary to the regular trap code, which appends $2t$ trap qubits, the strong trap code only appends a single $\ket0$ trap and a single $\ket+$ trap. These two traps are subsequently encoded using a quantum error-correcting code that has the desired properties described above, resulting in a ciphertext of the same length as the original trap code.

\begin{definition}[Strong trap code]
    Let $(E_i)_{i \in \mathbb{N}}$ be a weight-sparse family of weakly self-dual CSS codes with parameters $[[n(i),1,d(i) = \Omega(\sqrt{n(i)}]]$ and benign distance $\Omega(\sqrt{n(i)})$. Then the $i$th strong trap code $\{V_{i,k}\}_{k \in \Key_i}$ encodes $m = 1$ qubit using $t = 3n(i) - 1$ tags with the unitaries $V_{i,k} := \pi_k E_i^{\otimes 3} \H_{2n(i)+1}$ (where $\H_{2n(i)+1} = \I^{\otimes 2n(i)} \otimes \H \otimes \I^{\otimes (n(i)-1)}$).
\end{definition}
The strong trap code invokes two layers of security: the CSS codes $E_i$, which detect low-weight attacks, and the traps $\ket0$ and $\ket+$, which detect higher-weight attacks by revealing bit and phase flips, respectively. 

One can verify that computing on quantum states authenticated with the strong trap code works in much the same way as for the original trap code. For details, see~\cite{BGS13}.\footnote{For some applications, authenticating through measurement (cf.~\cite[Appendix B.2] {BGS13}) can be very useful. Our underlying code has all the requirements to achieve this in principle, but in this work we focus on quantum authentication and do not formulate the full security notions needed to properly describe this scenario.}

\begin{theorem}
The strong trap code is a strong-purity-testing code with error $\negl(n(i))$.\label{thm:STC-is-SPT}
\end{theorem}

\begin{proof}
Consider an arbitrary $i$ and non-identity Pauli $P_{\ell} \in \Pauli_{3n(i)} \backslash \{\I^{\otimes 3n(i)}\}$. Let $w_x$ and $w_z$ denote the $\X$-weight and $\Z$-weight (respectively) of $P_{\ell}$, and note that $\max(w_x,w_z) > 0$.

We bound the probability that $P_{\ell'} := \pi_k^{\dag} P_{\ell} \pi_k$ remains undetected by the code $E_i$ and the traps. Because $E_i$ is a CSS code, it detects $\X$ and $\Z$ errors separately: let us write $P_{\ell'} = P_xP_z$ with $P_x \in \{\I,\X\}^{\otimes 3n(i)}$ and $P_z \in \{\I,\Z\}^{\otimes 3n(i)}$, and focus first on the probability that $P_x$ remains undetected, i.e., the probability that $\H_{2n(i)+1}(E_i^{\dag})^{\otimes 3}P_{x}E_i^{\otimes 3}\H_{2n(i) + 1} \in \Pauli_1 \otimes \{\I,\Z\}^{\otimes 3n(i) - 1}$.

Because of the permutation $\pi_k$, $P_x$ is a random Pauli in $\{\I,\X\}^{\otimes 3n(i)}$ with weight $w_x$. (Note that $P_z$ is also a random Pauli with weight $w_z$, but is correlated with $P_x$: any overlap in the locations of $\X$ and $\Z$ operators in $P_{\ell}$ is preserved by the permutation.)

Consider all possible values of $w_x = w_1 + w_2 + w_3$, where $w_1$ denotes the weight of $P_x$ on the first (data) codeword, $w_2$ the weight on the second ($\ket0$-trap) codeword, and $w_3$ the weight on the third ($\ket+$-trap) codeword:
\begin{itemize}
    \item If $w_x = 0$, then the Pauli $P_x$ is identity, and remains undetected with probability 1.
    \item If $0 < w_x < d(i)$, then $0 < w_j < d(i)$ for at least one $j \in \{1,2,3\}$. $E_i$ detects an error on the $j$th block with certainty, since the weight of the error is below the distance and the benign distance.
    \item If $d(i) \leq w_x \leq 3n(i) - d(i)$, the attack $P_x$ will likely be detected on the second block, the $\ket0$-trap. We can be in one of four cases:
    \begin{itemize}
        \item $w_2 > 0$ and $P_x$ is detected in the second block by the CSS code $E_i$.
        \item $w_2 > 0$ and $P_x$ acts as a logical operation on the second block. Since $P_x$ consists of only $\I$'s and $\X$'s, this logical operation can only be an $\X$ by the construction of CSS codes. In this case, $P_x$ is detected by the projection that checks whether the trap is still in the $\ket0$ state.
        \item $w_2 > 0$ and $P_x$ acts as a stabilizer on the second block, and remains undetected on that block. However, by the weight-sparsity of the code family, the probability that this is the case is negligible in $n(i)$.
        \item $w_2 = 0$. In this case, $P_x$ acts as identity on the second block. The probability that this case occurs, however, is small:
        \begin{align}
            \Pr_k[w_2 = 0] = \frac{\binom{2n(i)}{w_x}}{\binom{3n(i)}{w_x}} < \left(\frac{2}{3}\right)^{w_x} \leq \left(\frac{2}{3}\right)^{d(i)} \ .
        \end{align}
        The first inequality holds in general for binomials, and the second one follows from the fact that $w_x \geq d(i)$. Since $d(i) = \Omega(\sqrt{n(i)})$, this probability is negligible in $n(i)$.
    \end{itemize}
    In total, the probability of the attack remaining undetected for $d(i) \leq w_x \leq 3n(i) - d(i)$ is negligible in $n(i)$.
    \item If $3n(i) - d(i) < w_x < 3n(i)$: as in the second case, there is at least one $j \in \{1,2,3\}$ such that $n(i) - d(i) < w_j < n(i)$, causing the attack to be detected (recall that $\X^{\otimes 3n(i)}$ is a logical $\X$, and therefore this mirrors the $0 < w_x < d(i)$ case).
    \item If $w_x = 3n(i)$, then the logical content of the second block, the $\ket0$-trap, is flipped. This is detected with certainty as well.
\end{itemize}
We see that unless $w_x = 0$, the Pauli $P_x$ remains undetected only with probability negligible in $n(i)$. A similar analysis can be made for $P_z$: it is always detected with high probability, unless $w_z = 0$. We stress that these probabilities are \emph{not} independent. However, we can say that
\begin{align}
    \Pr_k [ P_x \text{ and } P_z \text{ undetected}] \leq \min\left\{\Pr_k[ P_x \text{ undetected}], \ \Pr_k[ P_z \text{ undetected}]\right\},
\end{align}
and since at least one of $w_x$ and $w_z$ is non-zero, this probability is negligible in $n(i)$.
\qed
\end{proof}

%%%%%%%%%%%%%%%%%%%%%%%%%%%%%%%%%%%%%%%%%%%%%%%%%%%%%%%%%%%%%%%
\section{Simultaneous encryptions with key reuse}\label{sec:parallel-encryptions}
%%%%%%%%%%%%%%%%%%%%%%%%%%%%%%%%%%%%%%%%%%%%%%%%%%%%%%%%%%%%%%%

Earlier work on key reuse for quantum authentication deals explicitly with \emph{key recycling}, the decision to reuse (part of) a key for a new encryption after completing the transmission of some other quantum message. The key is reused only \emph{after} the honest party decides whether to accept or reject the first message, so recycling is a strictly sequential setting.

If Construction~\ref{con:qecc-to-skqes} is instantiated with a strong-purity-testing code (such as the strong trap code), the resulting scheme is able to handle an even stronger, parallel, notion of key reuse. As long as the one-time pads are independent, it is possible to encrypt multiple qubits under the same code key while preserving security. Even if the adversary is allowed to interactively decrypt a portion of the qubits one-by-one, the other qubits will remain authenticated. This property is especially important for the strong trap code: computing on data authenticated with the strong trap code requires all qubits to be encrypted under the same permutation key.

The original trap code is secure in this setting (as long as the one-time pads are fresh; see Section~5.2 of~\cite{BGS13}), but only if all qubits are decrypted at the same time. If some qubits can be decrypted separately, the adversary can deduce the location of the $\ket+$ traps by applying single-qubit $\X$ operations to different ciphertexts at different locations, and observing which ones are rejected. Repeating this for the $\Z$ operator to learn about the $\ket0$ traps, the adversary can completely break the authentication on the remaining qubits.

Suppose we encrypt two messages using an authentication scheme based on a strong-purity-testing code $\{V_{k_0}\}_{\Key_0}$, using the same code key $k_0$ but a fresh one-time pad. If we then decrypt the first message, the scheme is still $\QCAR$-authenticating on the second message with only slightly worse security.

\begin{theorem}[informal]\label{thm:parallel-key-reuse}
Let $(\Encrypt,\Decrypt)$ be an $\epsilon$-$\QCAR$-authenticating scheme resulting from Construction~\ref{con:qecc-to-skqes}, using a strong-purity-testing code $\{V_{k_0}\}_{\Key_0}$. Let $M_1, M_2$ denote the plaintext registers of the two messages, $C_1 = M_1T_1, C_2 = M_2T_2$ the corresponding ciphertext registers, and $R$ a side-information register. Let $\advA_1$, $\advA_2$ be arbitrary adversarial channels. Consider the setting where the adversary acts on the qubits, encrypted with keys $k_0, k_1, k_2$, as 
\[
\Decrypt^{\reg{C_2 \to M_2}}_{k_0,k_2} \circ \advA^{\reg{M_1,C_2,R}}_2 \circ \Decrypt^{\reg{C_1 \to M_1}}_{k_0,k_1} \circ \advA^{\reg{C_1,C_2,R}}_1 \circ \left(\Encrypt^{\reg{M_1 \to C_1}}_{k_0,k_1} \otimes \Encrypt^{\reg{M_2 \to C_2}}_{k_0,k_2} \right)\,,
\]
so that the key $k_0$ is used for both messages. Then, the scheme is  $2\epsilon$-$\QCAR$-authenticating on the second qubit.
\end{theorem}
\begin{proof}[Sketch]
As a first step, we rewrite the encryption of the second qubit as using encoding and teleportation, by using the equivalence between applying a random quantum one-time pad and teleporting a state. The encryption of the second qubit can then be thought of as happening after decryption of the first qubit. Next, we apply $\QCAR$ security of the first qubit, where we are using the property that $k_0$ is recycled both in the accept and the reject case. Finally we undo the rewrite and can directly apply $\QCAR$ security on the remaining state.
\qed
\end{proof}
The full proof can be found in Appendix~\ref{appendix:proof-parallel-key-reuse}. The argument easily extends to any polynomial number of authenticated qubits.

%%%%%%%%%%%%%%%%%%%%%%%%%%%%%%%%%%%%%%%%%%%%%%%%%%%%%%%%%%%%%%%
\section{Conclusion}\label{sec:conclusion}
%%%%%%%%%%%%%%%%%%%%%%%%%%%%%%%%%%%%%%%%%%%%%%%%%%%%%%%%%%%%%%%
We presented a new security definition, $\QCAR$, for ciphertext authentication with key recycling, and showed that schemes based on purity-testing codes satisfy quantum ciphertext authentication, while strong purity testing implies both ciphertext authentication and key recycling. This is  analogous to the security of quantum plaintext-authentication schemes from purity-testing codes~\cite{BCGST02,P17}. 

Additionally, we constructed the \emph{strong trap code}, a variant of the trap code which is a strong-purity-testing code and therefore is $\QCAR$ secure (as well as secure under all notions of plaintext authentication).
This new scheme can strengthen security and add key-recycling to earlier applications of the trap code. It is also applicable in a wider range of applications than the original trap code, because encrypted qubits remain secure even if other qubits sharing the same key are decrypted earlier.

A potential application of the strong trap code is the design of a quantum CCA2-secure encryption scheme (as in~\cite[Definition~9]{AGM17}) that allows for computation on the encrypted data. By only using the pseudo-random generator for the one-time-pad keys, and recycling the key for the underlying error-correcting code, this security level could be achieved.

As future work, our definition of $\QCAR$ could be generalized in different ways. First, one can consider a variant of the definition in the abstract-cryptography or universal-composability framework, in order to ease the composition with other cryptographic primitives. Second, because it can be useful to authenticate measurements in delegated-computation applications, one could extend the definition of $\QCAR$ to deal with the measurement of authenticated data. We expect no real obstacles for this extension of the definition, and refer to~\cite[Appendix~B.2]{BGS13} for comparable work on the original trap code.

%%%%%%%%%%%%%%%%%%%%%%%%%%%%%%%%%%%%%%%%%%%%%%%%%%%%%%%%%%%%%%%
\section*{Acknowledgements}\label{sec:acknowledgements}
%%%%%%%%%%%%%%%%%%%%%%%%%%%%%%%%%%%%%%%%%%%%%%%%%%%%%%%%%%%%%%%
We thank Gorjan Alagic, Christian Majenz, and Christian Schaffner for valuable discussions and useful input at various stages of this research. Aditionally, we thank Christian Schaffner for his comments on an earlier version of this manuscript.

Florian Speelman acknowledges financial support from the European Research Council (ERC Grant Agreement no 337603), the Danish Council for
Independent Research (Sapere Aude), Qubiz Quantum Innovation Center, and VILLUM FONDEN via the QMATH Centre of Excellence (Grant No.\ 10059).

%%%%%%%%%%%%%%%%%%%%%%%%%%%%%%%%%%%%%%%%%%%%%%%%%%%%%%%%%%%%%%%
\bibliography{SPT}
%%%%%%%%%%%%%%%%%%%%%%%%%%%%%%%%%%%%%%%%%%%%%%%%%%%%%%%%%%%%%%%

\appendix

%%%%%%%%%%%%%%%%%%%%%%%%%%%%%%%%%%%%%%%%%%%%%%%%%%%%%%%%%%%%%%%
\section{Proof of Theorem~\ref{thm:PT-implies-QCA}}\label{appendix:proof-of-PT-implies-QCA}
%%%%%%%%%%%%%%%%%%%%%%%%%%%%%%%%%%%%%%%%%%%%%%%%%%%%%%%%%%%%%%%
In this appendix, we work out the proof of Theorem~\ref{thm:PT-implies-QCA}. It follows the same strategy as the proof of Theorem~\ref{thm:SPT-implies-QCA-R} in Section~\ref{sec:QCA-R}, but the expressions generally take on a nicer form.

\begin{proof}
        Let $\advA$ be an adversary as in Definition~\ref{def:QCA}. Define a simulator $\simS$ in the same way as in the proof of Theorem~\ref{thm:SPT-implies-QCA-R}. Note that this simulator is of the required form in the accept case (see Definition~\ref{def:QCA}).
    
    We show that for this simulator, the distance $\frac{1}{2}\|\mathfrak{I} - \mathfrak{R}\|_{\diamond}$ between the ideal and the real channel is upper bounded by $\varepsilon$. Let $\rho^{\reg{MRE}}$ be any quantum state on the message register, side-information register, and an environment register $E$. Assume, as in the proof of Theorem~\ref{thm:SPT-implies-QCA-R}, that $\advA$ is a unitary map $U^{\reg{MTR}}$. Let $\Real_{k_1,k_2}$ and $\Ideal_{k_1,k_2}$ be the effective attacks in the real and ideal world, respectively:
    \begin{align}
        \Real_{k_1,k_2} &:= V^{\dag}_{k_1} P^{\dag}_{k_2}U^{\reg{MTR}}P_{k_2}^{\reg{MT}}V_{k_1}^{\reg{MT}} (\rho \otimes \proj{0^t}) V^{\dag}_{k_1} P^{\dag}_{k_2}U^{\dag}P_{k_2}V_{k_1},\\
    \Ideal_{k_1,k_2} &:= V^{\dag}_{k_1} P^{\dag}_{k_2}U^{\reg{M_1TR}}P_{k_2}^{\reg{M_1T}}V_{k_1}^{\reg{M_1T}}
        (\rho \otimes \proj{0^t, \Phi^+})
        V^{\dag}_{k_1} P^{\dag}_{k_2}U^{\dag}P_{k_2}V_{k_1}.
    \end{align}
    Then we can write the result of the real and the ideal channels as
    \begin{align}
        \mathfrak{R}(\rho) &= \expectation_{k_1,k_2} \left[\bra{0^t}^{\reg{T}} \Real_{k_1,k_2}\ket{0^t} \ + \  \proj{\bot}^{\reg{M}} \otimes \ \Tr_{M}\left(\sum_{i\neq 0^t} \bra{i}^{\reg{T}} \Real_{k_1,k_2}\ket{i}\right)\right],\\
        \mathfrak{I}(\rho) &= \expectation_{k_1',k_2'} \left[\bra{\Phi^+,0^t}^{\reg{M_1M_2T}} \Ideal_{k_1',k_2'}\ket{\Phi^+0^t} \ + \  \proj{\bot}^{\reg{M}} \otimes \ \Tr_{M}\left(\sum_{i\neq (\Phi^+,0^t)} \bra{i}^{\reg{M_1M_2T}} \Ideal_{k_1',k_2'}\ket{i}\right)\right].
    \end{align}
    
    These expressions are obtained simply by plugging in the description of the authentication scheme (see Construction~\ref{con:qecc-to-skqes}) and the simulator into the channels of Definition~\ref{def:QCA}. Since the accept states are orthogonal to the reject states in the $M$ register, the distance $\frac{1}{2}\|\mathfrak{I}(\rho) - \mathfrak{R}(\rho)\|_{\tr}$ can be written as
    \begin{align}
        &\frac{1}{2}\left\|\expectation_{k_1',k_2'} \left(\bra{\Phi^+,0^t}\Ideal_{k_1',k_2'}\ket{\Phi^+,0^t}\right) - \expectation_{k_1,k_2} \bra{0^t}\Real_{k_1,k_2}\ket{0^t}\right\|_{\tr} \nonumber\\
        &\quad + \ \ \frac{1}{2}\left\| \expectation_{k_1',k_2'} \left(\Tr_M \sum_{i \neq (0^t,\Phi^+)} \bra{i}\Ideal_{k_1',k_2'}\ket{i}\right) - \expectation_{k_1,k_2} \left(\Tr_M \sum_{i \neq 0^t} \bra{i} \Real_{k_1,k_2} \ket{i}\right)\right\|_{\tr} \ .\label{eq:distance-to-bound-QCA}
    \end{align}
    
    We can thus focus on bounding the two terms in equation~\eqref{eq:distance-to-bound-QCA}, for accept and reject, separately. Intuitively, the two states inside the first trace norm in equation~\eqref{eq:distance-to-bound-QCA} differ on those Paulis $P_{\ell}$ that are accepted in the real scenario, but not in the ideal one. The strong-purity-testing property promises that these Paulis are very few. We will work out this case; the second (the reject case) is similar.
    
    Decompose the attack as $U^{\reg{MTR}} = \sum_{\ell} \alpha_{\ell} P_{\ell}^{\reg{MT}} \otimes U_{\ell}^{\reg{R}}$. Rewrite the real accept case as
    \begin{align}
        &\phantom{= \ } \expectation_{k_1,k_2} \bra{0^t}\Real_{k_1,k_2}\ket{0^t}\\
        &= \expectation_{k_1,k_2} \bra{0^t} V^{\dag}_{k_1} P^{\dag}_{k_2}\left(\sum_{\ell} \alpha_{\ell} P_{\ell}^{\reg{MT}} \otimes U_{\ell}^{\reg{R}}\right)P_{k_2}^{\reg{MT}}V_{k_1}^{\reg{MT}} (\rho \otimes \proj{0^t}) V^{\dag}_{k_1} P^{\dag}_{k_2}\left(\sum_{\ell'} \alpha_{\ell'}^* P_{\ell'} \otimes U_{\ell'}^{\dag}\right)P_{k_2}V_{k_1}\ket{0^t}\\
        &= \expectation_{k_1,k_2} \sum_{\ell,\ell'} \alpha_{\ell}\alpha_{\ell'}^* \bra{0^t} \left(V^{\dag}_{k_1} P^{\dag}_{k_2}P_{\ell} P_{k_2}V_{k_1} \otimes U_{\ell}^{\reg{R}}\right) (\rho \otimes \proj{0^t}) \left(V^{\dag}_{k_1} P^{\dag}_{k_2} P_{\ell'}P_{k_2}V_{k_1} \otimes U_{\ell'}^{\dag}\right)\ket{0^t}.
    \end{align}
    By the Pauli twirl~\cite{DCEL09} (or by commuting and resolving the symplectic product, as is done in Appendix~\ref{appendix:proof-QCAR-derivation-real-reject}), this last line equals
    \begin{align}
        &\phantom{= \ } \expectation_{k_1} \sum_{\ell} \abs{\alpha_{\ell}}^2 \bra{0^t} \left(V^{\dag}_{k_1} P_{\ell} V_{k_1} \otimes U_{\ell}^{\reg{R}}\right) (\rho \otimes \proj{0^t}) \left(V^{\dag}_{k_1} P_{\ell}V_{k_1} \otimes U_{\ell}^{\dag}\right)\ket{0^t}\\
        &= \expectation_{k_1} \sum_{\ell : V_{k_1}^{\dag}P_{\ell}V_{k_1} \in \Pauli_{\real}} \abs{\alpha_{\ell}}^2 \left(Q_{k_1,\ell}^{\reg{M}} \otimes U_{\ell}^{\reg{R}}\right) \rho \left(Q_{k_1,\ell} \otimes U_{\ell}^{\dag}\right),
    \end{align}
    where $Q_{k_1,\ell}$ is the effective Pauli on the message register, induced by $V_{k_1}^{\dag}P_{\ell}V_{k_1}$, and where $\Pauli_{\real} := \Pauli_m \otimes \{\I,\Z\}^{\otimes t}$ is the set of (effective) Paulis that are undetected in the real scenario.
    
    With the same techniques, we can rewrite the ideal accept case as
    \begin{align}
        \expectation_{k_1',k_2'} \bra{\Phi^+,0^t} \Ideal_{k_1',k_2'} \ket{\Phi^+, 0^t}\ \  =\ \  \expectation_{k_1'} \sum_{\ell : V_{k_1'}^{\dag} P_{\ell} V_{k_1} \in \Pauli_{\ideal}} \abs{\alpha_{\ell}}^2 (Q_{k_1,\ell}^M \otimes U_{\ell}^{\reg{R}}) \rho (Q_{k_1,\ell}^M \otimes U_{\ell}^{\reg{R}}),
    \end{align}
    where $\Pauli_{\ideal} := \{\I^{\otimes m}\} \otimes \{\I,\Z\}^{\otimes t}$ is the set of Paulis that are undetected in the real scenario. Note that for $k_1$ and $\ell$ such that $V_{k_1}^{\dag}P_{\ell}V_{k_1} \in \Pauli_{\ideal}$, $Q_{\ell,k_1} = \I^{\otimes m}$.
    
    The distance between the ideal and the real accept states is thus
    \begin{align}
    &\phantom{= \ } \frac{1}{2}\left\|\expectation_{k_1',k_2'} \left(\bra{\Phi^+,0^t}\Ideal_{k_1',k_2'}\ket{\Phi^+,0^t}\right) - \expectation_{k_1,k_2} \bra{0^t}\Real_{k_1,k_2}\ket{0^t}\right\|_{\tr}\\
    &= \frac{1}{2} \left\| \expectation_{k_1} \sum_{\ell : V_{k_1}^{\dag} P_{\ell} V_{k_1} \in \Pauli_{\real} \backslash \Pauli_{\ideal}} \abs{\alpha_{\ell}}^2 (Q_{k_1,\ell}^M \otimes U_{\ell}^{\reg{R}}) \rho (Q_{k_1,\ell}^M \otimes U_{\ell}^{\reg{R}}) \right\|_{\tr}\\
    &\leq \frac{1}{2} \expectation_{k_1} \sum_{\ell : V_{k_1}^{\dag} P_{\ell} V_{k_1} \in \Pauli_{\real} \backslash \Pauli_{\ideal}} \abs{\alpha_{\ell}}^2,
    \end{align}
    by the triangle inequality for the trace norm. Let 
    $\delta_{(P_{a} \in A)}$ be 1 whenever $P_a \in A$, and 0 otherwise. Note that for all $k_1$,  $V_{k_1}^{\dag} P_{0} V_{k_1}  = \I^{\otimes (m+t)} \not\in \Pauli_{\real} \backslash \Pauli_{\ideal}$. This justifies the continuation of the derivation:
    \begin{align}
    &= \frac{1}{2} \sum_{\ell} \expectation_{k_1} \delta_{(V_{k_1}^{\dag} P_{\ell} V_{k_1} \in \Pauli_{\real} \backslash \Pauli_{\ideal})} \abs{\alpha_{\ell}}^2\\
    &= \frac{1}{2}\sum_{\ell \neq 0} \expectation_{k_1} \delta_{(V_{k_1}^{\dag} P_{\ell} V_{k_1} \in \Pauli_{\real} \backslash \Pauli_{\ideal})} \abs{\alpha_{\ell}}^2\\
    &\leq  \frac{1}{2} \sum_{\ell \neq 0} \expectation_{k_1} \delta_{(V_{k_1}^{\dag} P_{\ell} V_{k_1} \in \Pauli_{\real})} \abs{\alpha_{\ell}}^2\\
    &\leq \frac{1}{2} \sum_{\ell \neq 0} \epsilon \abs{\alpha_{\ell}}^2 \ \ \leq \ \ \frac{\epsilon}{2}.
    \end{align}
    The first inequality is because $\Pauli_{\real} \backslash \Pauli_{\ideal} \subseteq \Pauli_{\real}$, and the second inequality is by strong purity testing of the code $\{V_{k_1}\}_{k_1 \in \Key_1}$. This concludes the proof that in the accept case, the real and the ideal scenarios are $\epsilon$-close (i.e., the first term of equation~\eqref{eq:distance-to-bound-QCA} is upper bounded by $\sfrac{\epsilon}{2}$). The reject case is completely analogous.
    
    Summing the accept and reject case as in equation~\eqref{eq:distance-to-bound-QCA}, we see that $\frac{1}{2}\|\mathfrak{I}(\rho) - \mathfrak{R}(\rho)\|_{\tr} \leq \epsilon$ for all $\rho$. This concludes the proof.
    \qed
\end{proof}

%%%%%%%%%%%%%%%%%%%%%%%%%%%%%%%%%%%%%%%%%%%%%%%%%%%%%%%%%%%%%%%
\section{Details for the proof of Theorem~\ref{thm:SPT-implies-QCA-R}}
%%%%%%%%%%%%%%%%%%%%%%%%%%%%%%%%%%%%%%%%%%%%%%%%%%%%%%%%%%%%%%%

\subsection{Derivation of equation~\eqref{eq:distance-to-bound}}\label{appendix:proof-QCAR-derivation-distance}
We give details on how to arrive at equation~\eqref{eq:distance-to-bound} given the expressions for $\mathfrak{I}$ and $\mathfrak{R}$ in the proof of Theorem~\ref{thm:SPT-implies-QCA-R}.

\begin{align}
    &\phantom{= \ }\frac{1}{2}\left\|\mathfrak{I}(\rho) - \mathfrak{R}(\rho)\right\|_{\tr}\\
    &= \frac{1}{2}\Bigg\| \expectation_{k_1',k_2'} \left[\bra{\Phi^+,0^t}\Ideal_{k_1',k_2'}\ket{\Phi^+,0^t} \otimes \tau_{\Key} \right] + \expectation_{k_1',k_2'} \left[\proj{\bot} \otimes \ \Tr_M \left(\sum_{i \neq (\Phi^+,0^t)} \bra{i}\Ideal_{k_1',k_2'} \ket{i}\right) \otimes \tau_{\Key_1}\right]\nonumber\\
    &\quad - \expectation_{k_1,k_2} \left[\bra{0^t} \Real_{k_1,k_2} \ket{0^t} \otimes \proj{k_1k_2} \right] - \expectation_{k_1,k_2} \left[\proj{\bot} \otimes \ \Tr_{M} \left(\sum_{i \neq 0^t} \bra{i} \Real_{k_1,k_2} \ket{i}\right) \otimes \proj{k_1}\right]\Bigg\|_{\tr}\\
    &=\frac{1}{2}\left\| \expectation_{k_1',k_2'} \left[\bra{\Phi^+,0^t}\Ideal_{k_1',k_2'}\ket{\Phi^+,0^t} \otimes \tau_{\Key}\right] - \expectation_{k_1,k_2} \left[\bra{0^t} \Real_{k_1,k_2} \ket{0^t} \otimes \proj{k_1k_2} \right] \right\|_{\tr} + \nonumber\\
    &\quad \frac{1}{2}\left\| \expectation_{k_1',k_2'} \left[\Tr_M \left(\sum_{i \neq (\Phi^+,0^t)} \bra{i}\Ideal_{k_1',k_2'} \ket{i}\right) \otimes \tau_{\Key_1}\right] - \expectation_{k_1,k_2} \left[\Tr_{M} \left(\sum_{i \neq 0^t} \bra{i} \Real_{k_1,k_2} \ket{i}\right) \otimes \proj{k_1}\right] \right\|_{\tr},
\end{align}
because $\|\rho + \sigma\|_{\tr} = \|\rho\|_{\tr} + \|\sigma\|_{\tr}$ whenever $\rho$ and $\sigma$ are orthogonal (and the accept and reject states are orthogonal in the $M$ register), and because $\left\| \proj{a} \otimes \rho \right\|_{\tr} = \left\| \rho \right\|_{\tr}$ for basis states $\ket{a}$ (this allows us to get rid of the $\proj{\bot}$).

Using the same techniques, and observing that the $\proj{k_1k_2}$ states (or $\proj{k_1}$ in the case of reject) are all orthogonal to each other, we obtain equation~\eqref{eq:distance-to-bound}.

\subsection{Derivation of equation~\eqref{eq:real-reject-rewrite}}\label{appendix:proof-QCAR-derivation-real-reject}
We give details on how the equality in equation~\eqref{eq:real-reject-rewrite} is derived in the proof of Theorem~\ref{thm:SPT-implies-QCA-R}.
\begin{align}
    &\phantom{= \ }\expectation_{k_2}\left(\Tr_M \sum_{i \neq 0^t} \bra{i} \Real_{k_1,k_2} \ket{i}\right)\\
    &= \expectation_{k_2}\Bigg(\Tr_M \sum_{i \neq 0^t} \bra{i} V^{\dag}_{k_1} P^{\dag}_{k_2}\left(\sum_{\ell} \alpha_{\ell} P_{\ell}^{\reg{MT}} \otimes U_{\ell}^{\reg{R}}\right)P_{k_2}^{\reg{MT}}V_{k_1}^{\reg{MT}} (\rho^{\reg{MRE}} \otimes \proj{0^t}^{\reg{T}}) \nonumber\\
    &\qquad V^{\dag}_{k_1} P^{\dag}_{k_2}\left(\sum_{\ell'} \alpha_{\ell'}^* P_{\ell'} \otimes U_{\ell'}^{\dag}\right)P_{k_2}V_{k_1} \ket{i}\Bigg)\\
    &= \expectation_{k_2}\Bigg(\Tr_M \sum_{i \neq 0^t} \sum_{\ell,\ell'} \alpha_{\ell}\alpha_{\ell'}^* \bra{i} \left(V^{\dag}_{k_1} P^{\dag}_{k_2}P_{\ell}^{\reg{MT}}P_{k_2}^{\reg{MT}}V_{k_1}^{\reg{MT}} \otimes U_{\ell}^{\reg{R}}\right) (\rho^{\reg{MRE}} \otimes \proj{0^t}^{\reg{T}}) \nonumber\\
    &\qquad \left(V^{\dag}_{k_1} P^{\dag}_{k_2} P_{\ell'}P_{k_2}V_{k_1} \otimes U_{\ell'}^{\dag}\right) \ket{i}\Bigg),
\end{align}
where the second equality is obtained by moving summation signs further out and rearranging the unitaries so that they are sorted according to the register they act on. Recall that the $\ell$ and $\ell'$ index an $(m+t)$-qubit Pauli, and thus consist of $2(m+t)$ bits: we can regard any such Pauli-index $a$ as $(x_{a},z_{a})$ where the two parts are $(m+t)$-bit strings describing the locations of the $\X$ and $\Z$ Paulis, respectively. Define the symplectic inner product $\SP{a}{b} := x_a \cdot z_b - x_b \cdot z_a$, which equals $1$ whenever the Paulis $P_a$ and $P_b$ commute, and $-1$ when they anti-commute. Furthermore defining $\oplus$ as the bitwise xor operation, We can then continue the derivation as follows:

\begin{align}
    &= \expectation_{k_2} \Bigg( \Tr_M \sum_{i \neq 0^t} \sum_{\ell,\ell'} \alpha_{\ell}\alpha_{\ell'}^* (-1)^{\SP{\ell \oplus \ell'}{k_2}} \bra{i} \left(V^{\dag}_{k_1} P_{\ell}V_{k_1}^{\reg{MT}} \otimes U_{\ell}^{\reg{R}}\right) (\rho^{\reg{MRE}} \otimes \proj{0^t}^{\reg{T}})  \left(V^{\dag}_{k_1} P_{\ell'}V_{k_1} \otimes U_{\ell'}^{\dag}\right) \ket{i} \Bigg)\\
    &= \Tr_M \sum_{i \neq 0^t} \sum_{\ell} \abs{\alpha_{\ell}}^2 \bra{i} \left(V^{\dag}_{k_1} P_{\ell}V_{k_1}^{\reg{MT}} \otimes U_{\ell}^{\reg{R}}\right) (\rho^{\reg{MRE}} \otimes \proj{0^t}^{\reg{T}})  \left(V^{\dag}_{k_1} P_{\ell}V_{k_1} \otimes U_{\ell}^{\dag}\right) \ket{i}.
\end{align}
The last equality follows by Equality (2) in~\cite{P17}. This permuting of Paulis and the cancellation of terms for $\ell \neq \ell'$ is also called the Pauli Twirl~\cite{DCEL09}.

If we now apply the projector on the tag register $T$, and observe that the projection preserves exactly those $\ell$ such that $V_{k_1}^{\dag}P_{\ell}V_{k_1} \not\in \Pauli_{\real}$ (i.e., those $P_{\ell}$ that are rejected), the right-hand side of equation~\eqref{eq:real-reject-rewrite} is obtained.

\subsection{Derivation of equation~\eqref{eq:QCAR-bound-accept-case}}\label{appendix:proof-QCAR-accept}
We give details on how to obtain equation~\eqref{eq:QCAR-bound-accept-case}, which bounds the difference between the accept and reject case in the real scenario.

For this derivation, it will be useful to purify $\rho^{\reg{MRE}}$ as $\sum_i \beta_i \ket{\psi_i}^{\reg{MREE'}}$, where $E'$ is an extra environment register for the purification. Note that $\ip{\psi_i}{\psi_j} = 0$ for $i \neq j$.

With the same strategy as in Appendix~\ref{appendix:proof-QCAR-derivation-real-reject}, except for the last step, we can rewrite the real state as
\begin{align}
    &\phantom{= \ }\bra{0^t} \Real_{k_1,k_2} \ket{0^t}\\
    &= \bra{0^t} V^{\dag}_{k_1} P^{\dag}_{k_2}\left(\sum_{\ell} \alpha_{\ell} P_{\ell}^{\reg{MT}} \otimes U_{\ell}^{\reg{R}}\right)P_{k_2}^{\reg{MT}}V_{k_1}^{\reg{MT}} \left(\sum_{i,i'}\beta_i\beta_{i'}^* \ket{\psi_i}\bra{\psi_{i'}} \ \otimes \ \proj{0^t}^{\reg{T}}\right) \nonumber\\
    &\qquad V^{\dag}_{k_1} P^{\dag}_{k_2}\left(\sum_{\ell'} \alpha_{\ell'}^* P_{\ell'} \otimes U_{\ell'}^{\dag}\right)P_{k_2}V_{k_1} \ket{0^t}\\
    &= \sum_{i,i'} \sum_{\ell,\ell'} \beta_i \alpha_{\ell}\beta_{i'}^* \alpha_{\ell'}^* \bra{0^t} \left(V^{\dag}_{k_1} P^{\dag}_{k_2}P_{\ell}P_{k_2}V_{k_1}^{\reg{MT}} \otimes U_{\ell}^{\reg{R}}\right) \ket{\psi_i}\ket{0^t}\bra{\psi_{i'}}\bra{0^t} \left(V^{\dag}_{k_1} P^{\dag}_{k_2} P_{\ell'}P_{k_2}V_{k_1} \otimes U_{\ell'}^{\dag}\right) \ket{0^t}\\
    &= \sum_{i,i'} \sum_{\ell,\ell'} \beta_i  \alpha_{\ell}\beta_{i'}^*\alpha_{\ell'}^* (-1)^{\SP{\ell}{k_2}}(-1)^{\SP{\ell'}{k_2}} \nonumber\\
    &\qquad \bra{0^t} \left(V^{\dag}_{k_1} P_{\ell}V_{k_1}^{\reg{MT}} \otimes U_{\ell}^{\reg{R}}\right) \ket{\psi_i}\ket{0^t}\bra{\psi_{i'}}\bra{0^t} \left(V^{\dag}_{k_1} P_{\ell'}V_{k_1} \otimes U_{\ell'}^{\dag}\right) \ket{0^t},
\end{align}
which, as a real state, can be expressed as
\begin{align}
    &\phantom{= \ }\sum_i \sum_{\ell} \beta_i \alpha_{\ell} (-1)^{\SP{\ell}{k_2}} \bra{0^t} \left(V^{\dag}_{k_1} P_{\ell}V_{k_1}^{\reg{MT}} \otimes U_{\ell}^{\reg{R}}\right) \ket{\psi_i}\ket{0^t}\\
    &= \sum_{\ell : V_{k_1}^{\dag}P_{\ell}V_{k_1} \in \Pauli_{\real}}  \alpha_{\ell} (-1)^{\SP{\ell}{k_2}} (Q_{k_1,\ell}^{\reg{M}} \otimes U_{\ell}^{\reg{R}}) \sum_i \beta_i \ket{\psi_i},
\end{align}
where $Q_{k_1,\ell}$ is the Pauli on the message register $M$ that results from $V_{k_1}^{\dag}P_{\ell}V_{k_1} \in \Pauli_{\real} = \Pauli_m \otimes \{\I,\Z\}^{\otimes t}$.
In much the same way, but using Equality (2) in~\cite{P17} as we did in Appendix~\ref{appendix:proof-QCAR-derivation-real-reject}, we can write the accept state in the ideal scenario,
\begin{align}
    \expectation_{k_1',k_2'} \left(\bra{\Phi^+,0^t}\Ideal_{k_1',k_2'}\ket{\Phi^+,0^t}\right),
\end{align}
as the pure state
\begin{align}
    &\phantom{= \ }\expectation_{k_1'} \left(\sum_{\ell: V_{k_1'}^{\dag}P_{\ell}V_{k_1'} \in \Pauli_{\ideal}} \alpha_{\ell} U_{\ell}^{\reg{R}} \sum_i \beta_i \ket{\psi_i}\right)\\
    &= \sum_{\ell} \expectation_{k_1'} \delta_{(V_{k_1'}^{\dag} P_{\ell} V_{k_1'} \in \Pauli_{\ideal})} \alpha_{\ell}U_{\ell}^{\reg{R}} \sum_i \beta_i \ket{\psi_i},
\end{align}
where $\delta_{(V_{k_1'}^{\dag} P_{\ell} V_{k_1'} \in \Pauli_{\ideal})}$ is the indicator function that is equal to 1 whenever $V_{k_1'}^{\dag} P_{\ell} V_{k_1'} \in \Pauli_{\ideal}$, and 0 otherwise. Continuing, we rewrite the ideal accept state as
\begin{align}
    &= \alpha_0 U_{0}^{\reg{R}} \sum_i \beta_i \ket{\psi_i} \ \ + \ \ \sum_{\ell \neq 0} \expectation_{k_1'} \delta_{(V_{k_1'}^{\dag} P_{\ell} V_{k_1'} \in \Pauli_{\ideal})} \alpha_{\ell}U_{\ell}^{\reg{R}} \sum_i \beta_i \ket{\psi_i}.\label{eq:appendix-accept-ideal}
\end{align}
(Recall that $P_0 = \I^{\otimes (m+t)}$ by convention.) By the strong-purity-testing property of the code $\{V_{k_1}\}_{k_1 \in \Key_1}$, and the fact that $\Pauli_{\ideal} \subseteq \Pauli_{\real}$, the second term in equation~\eqref{eq:appendix-accept-ideal} has very small amplitude. Thus, the ideal accept state is $\epsilon$-close in trace distance to $\alpha_0U_0^{\reg{R}} \sum_i \beta_i \ket{\psi_i}$.

We are now ready to bound the expected distance between the ideal and the real case:
\begin{align}
    &\phantom{= \ }\expectation_{k_1,k_2} \frac{1}{2} \left\|\expectation_{k_1',k_2'} \left(\bra{\Phi^+,0^t}\Ideal_{k_1',k_2'}\ket{\Phi^+,0^t}\right) - \bra{0^t}\Real_{k_1,k_2} \ket{0^t}\right\|_{\tr}\\
    &\leq \frac{\epsilon}{2} + \expectation_{k_1,k_2} \frac{1}{2} \Bigg\|\abs{\alpha_0}^2 \sum_{i,i'} \beta_i\beta_{i'} U_0\ket{\psi_i}\bra{\psi_{i'}}U_0^{\dag}  \ \ - \left(\sum_{\ell:V_{k_1}^{\dag}P_{\ell}V_{k_1} \in \Pauli_{\real}} \alpha_{\ell} (-1)^{\SP{\ell}{k_2}} (Q_{k_1,\ell} \otimes U_{\ell}) \sum_i \beta_i \ket{\psi_i}\right)\nonumber\\
    &\quad \left(\sum_{\ell':V_{k_1}^{\dag}P_{\ell'}V_{k_1} \in \Pauli_{\real}} \alpha_{\ell'} (-1)^{\SP{\ell'}{k_2}} \sum_i \beta_i \bra{\psi_i}(Q_{k_1,\ell'} \otimes U_{\ell'}^{\dag}) \right)\Bigg\|_{\tr}\\
    &\leq \frac{\epsilon}{2} + \expectation_{k_1,k_2} \left\| \alpha_0U_0 \sum_i\beta_i \ket{\psi_i} - \left(\sum_{\ell:V_{k_1}^{\dag}P_{\ell}V_{k_1} \in \Pauli_{\real}} \alpha_{\ell} (-1)^{\SP{\ell}{k_2}} (Q_{k_1,\ell} \otimes U_{\ell}) \sum_i\beta_i \ket{\psi_i}\right)\right\|.
\end{align}
The first inequality is the triangle inequality, and the second one follows from Lemma C.2 in~\cite{P17}, which states that $\frac{1}{2}\|\proj{\phi} - \proj{\psi}\|_{\tr} \ \leq \|\ket{\phi} - \ket{\psi}\|$ (where the latter is the vector 2-norm, i.e., $\|\ket{\phi}\| := \sqrt{\ip{\phi}{\phi}}$).

Continuing,
\begin{align}
    &= \frac{\epsilon}{2} + \expectation_{k_1,k_2} \left\| \sum_i\beta_i \sum_{\ell \neq 0:V_{k_1}^{\dag}P_{\ell}V_{k_1} \in \Pauli_{\real}} \alpha_{\ell} (-1)^{\SP{\ell}{k_2}} (Q_{k_1,\ell} \otimes U_{\ell}) \ket{\psi_i}\right\|\\
    &= \frac{\epsilon}{2} + \expectation_{k_1,k_2} \sqrt{\sum_{i,i'} \beta_i^*\beta_{i'}\sum_{\ell \neq 0:V_{k_1}^{\dag}P_{\ell}V_{k_1} \in \Pauli_{\real}} \sum_{\ell' \neq 0:V_{k_1}^{\dag}P_{\ell'}V_{k_1} \in \Pauli_{\real}} \alpha_{\ell}^*\alpha_{\ell'} (-1)^{\SP{\ell \oplus \ell'}{k_2}} \bra{\psi_i} (Q_{k_1,\ell} \otimes U_{\ell}^{\dag})(Q_{k_1,\ell'} \otimes U_{\ell}) \ket{\psi_{i'}}}\\
    &\leq \frac{\epsilon}{2} + \sqrt{\expectation_{k_1,k_2} \sum_{i,i'} \beta_i^*\beta_{i'} \sum_{\ell \neq 0:V_{k_1}^{\dag}P_{\ell}V_{k_1} \in \Pauli_{\real}} \sum_{\ell' \neq 0:V_{k_1}^{\dag}P_{\ell'}V_{k_1} \in \Pauli_{\real}} \alpha_{\ell}^*\alpha_{\ell'} (-1)^{\SP{\ell \oplus \ell'}{k_2}} \bra{\psi_i} (Q_{k_1,\ell} \otimes U_{\ell}^{\dag})(Q_{k_1,\ell'} \otimes U_{\ell}) \ket{\psi_{i'}}}\\
    &= \frac{\epsilon}{2} + \sqrt{\expectation_{k_1} \sum_{i,i'} \beta_i^*\beta_{i'} \sum_{\ell \neq 0:V_{k_1}^{\dag}P_{\ell}V_{k_1} \in \Pauli_{\real}} \abs{\alpha_{\ell}}^2 \bra{\psi_i} (Q_{k_1,\ell} \otimes U_{\ell}^{\dag})(Q_{k_1,\ell} \otimes U_{\ell}) \ket{\psi_{i'}}},
\end{align}
by Jensen's inequality and by Equation (2) in~\cite{P17}. This last line we can simplify greatly to

\begin{align}
    &= \frac{\epsilon}{2} + \sqrt{\expectation_{k_1} \sum_i\abs{\beta_i}^2 \sum_{\ell \neq 0:V_{k_1}^{\dag}P_{\ell}V_{k_1} \in \Pauli_{\real}} \abs{\alpha_{\ell}}^2 \ip{\psi_i}{\psi_i} }\\
    &= \frac{\epsilon}{2} + \sqrt{\expectation_{k_1} \sum_{\ell \neq 0:V_{k_1}^{\dag}P_{\ell}V_{k_1} \in \Pauli_{\real}} \abs{\alpha_{\ell}}^2 }\\
    &= \frac{\epsilon}{2} + \sqrt{\sum_{\ell \neq 0} \expectation_{k_1} \delta_{(V_{k_1}^{\dag}P_{\ell}V_{k_1} \in \Pauli_{\real})} \abs{\alpha_{\ell}}^2 }\\
    &\leq \frac{\epsilon}{2} + \sqrt{\sum_{\ell \neq 0} \epsilon \abs{\alpha_{\ell}}} \quad \leq \quad  \frac{\epsilon}{2} + \sqrt{\epsilon}.
\end{align}
The inequality is again due to the strong-purity-testing property of the code $\{V_{k_1}\}_{\Key_1}$. This concludes the derivation of equation~\eqref{eq:QCAR-bound-accept-case}.

%%%%%%%%%%%%%%%%%%%%%%%%%%%%%%%%%%%%%%%%%%%%%%%%%%%%%%%%%%%%%%%
\section{A high-benign-distance, weight-sparse code family}\label{appendix:construction-of-code-family}
%%%%%%%%%%%%%%%%%%%%%%%%%%%%%%%%%%%%%%%%%%%%%%%%%%%%%%%%%%%%%%%
In this appendix, we construct a family of quantum error-correcting codes that has the properties required for the construction of the strong trap code. The family will consist of weakly self-dual CSS codes that are based on classical Reed--Muller codes. Before we give the construction of the code family itself, we prove useful lemmas about the benign distance and weight sparsity of CSS codes.

\subsection{Benign distance of CSS codes}\label{appendix:high-benign-distance}
This first lemma shows that if a (weakly self-dual) CSS code has high distance, then also its benign distance must be high.
	
\begin{lemma}\label{lem:css-benign-distance}
	Let $C_1$ an $C_2$ be two classical linear codes such that $C_2 \subset C_1$ and $C_1^{\perp} = C_2$. Then the benign distance of $CSS(C_1,C_2)$ is greater than or equal to its conventional distance.
\end{lemma}
	
\begin{proof}
	Let $d$ denote the distance of the classical code $C_1$. By the construction of CSS codes, $CSS(C_1,C_2)$ also has (conventional) distance $d$.
	
	Also by construction, the check matrix of $CSS(C_1,C_2)$ is given by
	\begin{align}
	\left[
	\begin{array}{c c}
	H(C_2^{\perp}) & 0\\
	0 & H(C_1)
	\end{array}
	\right]
	=
	\left[
	\begin{array}{c c}
	G(C_2) & 0\\
	0 & G(C_2)
	\end{array}
	\right],
	\end{align}
	where $H(\cdot)$ represents the parity check matrix of a classical code, and $G(\cdot)$ the generator matrix.
	
	The rows of $G(C_2)$ form a basis for the codewords in $C_2$. Since $C_1$ has distance $d$, and $C_2 \subset C_1$, any row in $G(C_2)$, and any linear combination of these rows, has weight at least $d$. Thus, any linear combination of rows in the above check matrix also has weight at least $d$.
	
	The rows of the check matrix generate the stabilizers of the code~\cite{NC00}. We may conclude that the stabilizers of the code $CSS(C_1, C_2)$ all have weight at least $d$, and therefore the benign distance of the code is at least $d$.
	\qed
\end{proof}

We now further specify our construction: if the CSS code is built from a so-called \emph{punctured} classical code, then its distance (and therefore benign distance) is high.

\begin{lemma}\label{lem:css-selfdual-benign-distance}
		Let $C$ be an $[n,m,d]$ self-dual linear code for some $d > 1$. Assume w.l.o.g.\footnote{If they do, pick a different position to puncture at. Since $d \neq 0$, such a position always exists.} that not all codewords in $C$ end in 0. Define
		\[
		C_1 := \{c \in \{0,1\}^{n-1} \mid c 0 \in C \vee c1 \in C\}
		\]
		and $C_2 := C_1^{\perp}$. Then $CSS(C_1, C_2)$ is an $[[n-1,1, d']]$ code with $d' \in \{d-1,d\}$ and with benign distance $d_b \geq d' \geq d-1$.
	\end{lemma}

\begin{proof}
	    The code $C_1$ is a $[n-1,m,d']$ code for some $d' \in \{d-1,d\}$. Firstly, it has length $n-1$, because one bit is removed (punctured) from the codewords of $C$. Secondly, it has rank $m$: there are no two codewords in $C$ that differ at only the punctured bit (since $d > 1$), and so the punctured versions of two distinct codewords are also distinct. Thirdly, it has distance $d-1$, again because one bit is removed from the codewords in $C$ which all had weight at least $d$.
	    
	    In order to prove the statement of the lemma, we need to show two things: firstly, that $C_2$ has rank $m-1$ (from which the parameters of the CSS code will follow), and secondly, that $C_2 \subset C_1$ (from which $d_b \geq d'$ will follow by Lemma~\ref{lem:css-benign-distance}, and which shows that $C_1$ and $C_2$ are valid candidates for the CSS construction).
	    
	    We start by showing a stronger version of the second claim, namely that $C_2 = \{c_2 \mid c_20 \in C\}$. The latter set is then clearly a subset of $C_1$. For the forward inclusion, pick an arbitrary $c_2 \in C_2$. For all $c \in C$, which by definition of $C_1$ we can write as $c = c_1b$ for $c_1 \in C_1$ and $b \in \{0,1\}$, it follows that $\langle c_2 0, c\rangle = \langle c_2 0,c_1b\rangle = \langle c_2, c_1 \rangle + 0 = 0$. The last equality follows from the fact that $C_2 = C_1^{\perp}$. And so, $c_2 0 \in C^{\perp} = C$ (as $C$ is self-dual). For the other inclusion, pick a $c_2 \in \{ c_2 \mid c_20 \in C\}$. To show that $c_2 \in C_2 = C_1^{\perp}$, we show that $\langle c_1, c_2 \rangle = 0$ for all $c_1 \in C_1$: let $c_1 \in C_1$, and note that by definition, there is a $b \in \{0,1\}$ such that $c_1b \in C$. Then $\langle c_1, c_2 \rangle = \langle c_1b, c_20 \rangle = 0$, as $C$ is self-dual.
		
		It remains to show that $|C_2| = \frac{1}{2}|C_1|$, i.e., $C_2$ has rank $m-1$. For this, the stronger statement $C_2 = \{c \mid c0 \in C\}$ proven above will be useful. To see that $|\{c \mid c0 \in C\}| = \frac{1}{2}|C_1|$, consider a basis $\{v_1, ..., v_m\}$ for $C$, and let $I \subset [m]$ be the set of indices $i$ such that $(v_i)_n = 1$ (recall that we assumed without loss of generality that $I \neq \emptyset$). Note that all $\vec{x} \in \{0,1\}^m$ represent a (unique) codeword $c = x_1v_1 + x_2v_2 + \dots + x_mv_m$ of $C$, and conversely every codeword in $C$ is represented by some $\vec{x}$. write every $c \in C$ as $x_1v_1 + x_2v_2 + \cdots x_mv_m$ for some $\vec{x} \in \{0,1\}^m$. Since $I$ is non-empty, exactly half of all $\vec{x} \in \{0,1\}^m$ have $\sum_{i \in I} x_i = 0 \mod 2$ (resulting in the $n$th bit of $c$ being $0$), and exactly half have $\sum_{i \in I} x_i = 1 \mod 2$ (resulting in the $n$th bit of $c$ being $1$). Thus, exactly half of the elements in $C$ are of the form $a 0$ for some $a \in \{0,1\}^{n-1}$. Since $d > 1$, and so distinct codewords in $C$ are punctured to distinct codewords in $C_1$, the statement $|\{c \mid c0 \in C\}| = \frac{1}{2}|C_1|$ follows.
		
		We may conclude that the rank of $C_2$ is $m-1$. Thus, $CSS(C_1,C_2)$ is an $[[n-1,1,d']]$ code for $d' \in \{d-1,d\}$.
		\qed
	\end{proof}

\subsection{Weight sparsity of CSS codes}\label{appendix:weight-sparsity}
As it turns out, for CSS code families it suffices to show that the family is $\X$-weight sparse, as illustrated by the definition and lemma below.
	
	\begin{definition}[$\X$-weight-sparse code family]
		Let $(E_i)_{i\in\mathbb{N}}$ be a family of quantum error-correction codes with parameters $[[n(i), m(i), d(i)]]$, and define the sets $A_i(x,y,z)$ and $B_i(x,y,z)$ as in definition~\ref{def:weight-sparse}. Moreover, define $A_i^{\X}(x) := A_i(x,0,0)$ and $B_i^{\X}(x) := B_i(x,0,0)$.
		
		The family $(E_i)_{i \in \mathbb{N}}$ is \emph{$\X$-weight sparse} if the function
		\[
		f_{\X}(i) := \max_{w \leq n(i)} \frac{| B_i^{\X}(w) |}{ | A_i^{\X}(w) | }
		\]
		is negligible in $n(i)$.
	\end{definition}
	
	Note that for weakly self-dual CSS codes, $\X$-weight sparsity immediately implies $\Z$-weight sparsity (defined analogously), as weakly self-dual CSS codes are symmetric in their $\X$ and $\Z$ stabilizers. It also implies general weight sparsity:
	
	\begin{lemma}\label{lem:css-x-weight-to-weight}
		If the codes in the family $(E_i)_{i \in \mathbb{N}}$ are all CSS codes, and the family is $\X$-weight sparse, then the family is also weight sparse.
	\end{lemma}
	
	\begin{proof}
		Recall from the proof of Lemma~\ref{lem:css-benign-distance} that a CSS code has a stabilizer generating set containing elements that are built up of either exclusively $\X$ and $\I$ (we will call these $\X$-stabilizers), or exclusively $\Z$ and $\I$ (which we will call $\Z$-stabilizers). Thus, any stabilizer $P_{\ell}$ for the CSS code can be written as a product $P_{\ell_x}P_{\ell_z}$ of an $\X$-stabilizer $P_{\ell_x}$ and a $\Z$-stabilizer $P_{\ell_z}$.
		
		Consider a code $E_i$ in the family, and arbitrary nonnegative integers $x,y,z$ such that $x + y + z \leq n(i)$. Then elements of $A_i(x,y,z)$ can be constructed by first selecting a Pauli $P_{\ell_x}$ of the appropriate weight, and then selecting a Pauli $P_{\ell_z}$ of the appropriate weight and overlapping with $P_{\ell_x}$ at an appropriate number of positions:
		\begin{align}
		A_i(x,y,z) = \{P_{\ell_x}P_{\ell_z} \mid P_{\ell_x} \in A_i^{\X}(x+y) \mbox{ and } P_{\ell_z} \in A_i^{\Z}(y+z \mid P_{\ell_x}, y)\},
		\end{align}
		where $A_i^{\Z}(y+z \mid P_{\ell_x},y)$ denotes the subset of $A_i(0,0,x+y)$ that overlap with $P_{\ell_x}$ on exactly $y$ positions. Note that for all $P_{\ell_x} \in A_i^{\X}(x+y)$,
		\begin{align}
		|A_i^{\Z}(y+z \mid P_{\ell_x},y)| = |A_i^{\Z}(y+z \mid P_{\ell_0}, y)|,
		\end{align}
		for the canonical $P_{\ell_0} := \X^{\otimes(x+y)}\I^{\otimes(n(i)-x-y)}$. Hence,
		\begin{align}
		|A_i(x,y,z)| = |A_i^{\X}(x+y)| \cdot |A_i^{\Z}(y+z \mid P_{\ell_0}, y)|.
		\end{align}
		Similarly, define $B_i^{\Z}(y+z \mid P_{\ell_x},y)$ to be the benign subset of $A_i^{\Z}(y+z \mid P_{\ell_x},y)$. Using the same reasoning as above, we can arrive at the inequality
		\begin{align}
		|B_i(x,y,z)| = |B_i^{\X}(x+y)| \cdot |B_i^{\Z}(y+z \mid P_{\ell_0}, y)| \leq |B_i^{\X}(x+y)| \cdot |A_i^{\Z}(y+z \mid P_{\ell_0}, y)|.
		\end{align}
		Combining the above results, we conclude that for all nonnegative integers $x, y,$ and $z$ such that $x + y + z \leq n(i)$,
		\begin{align}
		\frac{|B_i(x,y,z)|}{|A_i(x,y,z)|} \leq \frac{|B_i^{\X}(x+y)| \cdot |A_i^{\Z}(y+z \mid P_{\ell_0}, y)|}{|A_i^{\X}(x+y)| \cdot|A_i^{\Z}(y+z \mid P_{\ell_0}, y)|} = \frac{|B_i^{\X}(x+y)|}{|A_i^{\X}(x+y)|}.
		\end{align}
		Maximizing over $x,y,z$ on both sides of the inequality, and noting that $x+y \leq n(i)$, the statement of the lemma follows.
		\qed
	\end{proof}

	\subsection{A high-benign-distance, weight-sparse QECC family}
In this subsection, we construct a family of quantum error-correcting codes that is based on classical Reed--Muller codes~\cite{Shor96,Pre97}. We show that it has distance and benign distance $O(\sqrt{n(i)})$, where $n(i)$ is the codeword length of the $i$th code in the family, and that it is weight sparse. 
	
Reed--Muller codes are a class of codes based on polynomials on the field $\mathds{F}^a_2$ for some $a \in \mathbb{N}$. Every polynomial on the field is associated with a vector of length $2^a$, representing the values of that polynomial on every possible input. The Reed--Muller code $R(i,a)$ consists of all these vectors for polynomials of degree up to $i$. Increasing $i$ while keeping $a$ constant results in a higher-ranked code, but with a smaller distance. For a discussion, see~\cite{Pre97}, Chapter 7.
	
For our purposes, we will be interested in Reed--Muller codes where $a$ scales with $i$ as $a = 2i+1$ (for $i \in \mathbb{N}$), as these codes happen to be self-dual. They have codeword length $2^{2i+1}$, rank $2^{2i}$, and distance $2^{i+1}$~\cite{Pre97}. Instantiating the construction in Lemma~\ref{lem:css-selfdual-benign-distance} with $R(i,2i+1)$, we see that the resulting quantum code $R_i$ is an $[[n(i),m(i),d(i)]]$ code with $n(i) = 2^{2i+1} - 1$, $m(i) = 1$, and $d(i) = \sqrt{2^{2(i+1)}} - 1$. In the resulting family $(R_i)_{i \in \mathbb{N}}$, the codeword length grows approximately quadratically with the desired distance (since $n(i) = \frac{1}{2}(d(i) + 1)^2+ 1)$. The benign distance is also high, at least $d(i) - 1$.
	
\begin{lemma}\label{lem:RM-code-weight-sparse}
	The family $(R_i)_{i \in \mathbb{N}}$ of quantum error-correction codes, where $R_i$ is constructed by puncturing the self-dual Reed--Muller code $R(i,2i+1)$, is weight sparse.
\end{lemma}
	
	\begin{proof}
		Since all the codes $R_i$ are CSS codes, by Lemma~\ref{lem:css-x-weight-to-weight} we only need to show that the family is $\X$-weight sparse, i.e.,
		
		\begin{align}
		\max_{w \leq n(i)} \frac{B_i^{\X}(w)}{A_i^{\X}(w)}
		\end{align}
		is negligible in $n(i)$. Recall from Lemma~\ref{lem:css-x-weight-to-weight} that $A_i^{\X}(w)$ refers to the set of Paulis consisting of $w$ Pauli-$\X$s and $n(i)-w$ identity operations, in any order. $B_i^{\X}$ is the benign subset of $A_i^{\X}$.
		
		The quantity $A_i^{\X}(w)$ can be expressed as $\binom{n(i)}{w}$. By construction of $R_i$, all its $\X$-stabilizers are generated by the generators of $C_{i,2}$, the dual of the punctured version of the classical Reed--Muller code $R(i,2i+1)$ (see also the proof of Lemma~\ref{lem:css-benign-distance}). Hence, the quantity $B_i^{\X}(w)$ is equal to $|C_{i,2}(w)|$, the number of codewords in $C_{i,2}$ of weight $w$. This in turn is upper bounded by $|C_{i,1}(w)|$, where $C_{i,1}$ is the punctured Reed--Muller code. ($C_{i,2}$ is the even subcode of $C_{i,1}$.) Since $1^{n(i)} \not\in C_{i,2}$ and so $|C_{i,2}(n(i))| = 0$, it suffices to show that
		\begin{align}
		\max_{w \leq n(i) - 1} \frac{|C_{i,1}(w)|}{\binom{n(i)}{w}}
		\end{align}
		is negligible in $n(i)$.

		First, note that since $1^{n(i)} \in C_{i,1}$, and so for every string $s \in C_{i,1}$ also $1^{n(i)} \oplus s \in C_{i,1}$, there are an equal amount of strings with weight $w$ in the code as with weight $n(i)-w$. Therefore, it suffices to consider only those $w \in \{1, \dots, (n(i)-1)/2 \}$. We find an upper bound to the above expression, for three separate cases:
		
		\begin{description}
			\item[$0 < w < d(i)$:]
			The bound follows directly from the error-correcting property of the code: there is no codeword in $C_{i,1}$ with weight less than the distance $d(i)$. Hence, $|C_{i,1}(w)|/\binom{n(i)}{w} \leq 0$ for these values of $w$.
			
			\item[$d \leq w < \frac{n(i)}{8}$:] For this case, we can use a theorem on
			intersecting sets by Ray-Chaudhuri and Wilson~\cite{RW75}. Consider two (non-identical) strings
			$s, t \in C_{i,1}$ with $|s|=|t|=w$. Then, since $C_{i,1}$ is a linear code,
			if we define $u=s \oplus t$, we have that also $u \in C_{i,1}$. In particular, $|u| \geq d$. Now write $|u| = |s| + |t| - 2 |s \wedge t| \geq d(i)$, and therefore we have
			\begin{align}
			|s \wedge t| \leq \frac{1}{2}(|s| + |t| - d(i)) = w - \frac{d(i)}{2} .
			\end{align}
			Now, instead of bitstrings, we view all strings in $C_{i,1}$ of weight $w$ as a family $\F_w$ of subsets of $[n(i)]$. To be more precise,
			define $ \F_w = \bigl\{ \{j : j \in [n(i)], s_j=1 \} : s \in C_{i,1}(w) \bigr\}$. % Yfke: we interpret the codewords of C_1 as characteristic strings for the sets in F_w
			We will bound the size of this set family, noting that $|\mathcal{F}_w| = |C_{i,1}(w)|$.
			From the previous argument, we have $\forall F,G \in \mathcal{F}_w, F\neq G : |F \cap G| \leq w - \frac{d(i)}{2}$.
			
			From the Ray-Chaudhuri--Wilson inequality, stated in Lemma~\ref{lem:RayChaudhuriWilson}, it then immediately follows that
			\begin{equation}
			|\mathcal{F}_w| \leq \binom{n(i)}{w - \frac{d(i)}{2}} .
			\end{equation}
			
			Now, to bound the ratio of this bound to the total number of strings of weight $w$.
			\begin{align}
			\frac{ |\mathcal{F}_w| }{ \binom{n(i)}{w} } & \leq \frac{ \binom{ n(i) }{ w - \frac{d(i)}{2} } }{ \binom{n(i)}{w} }\\
			&= \frac{ w! (n(i)-w)! }{ (w-\frac{d(i)}{2})! (n(i)-w+\frac{d(i)}{2})! } \\
			&= \frac{ w^{ \underline{ \frac{d(i)}{2} } } }{ (n(i)-w+\frac{d(i)}{2})^{\underline{\frac{d(i)}{2}}} } \\
			&\leq \frac{ \left(\frac{n(i)}{8}\right)^{ \underline{ \frac{d(i)}{2} } } }{  \left(\frac{7n(i)}{8}\right)^{\underline{\frac{d(i)}{2} }}  } 
			\leq \left( \frac{1}{ 7} \right) ^{ \frac{d(i)}{2} } .
			\end{align}
			Here $m^{ \underline{k} }$ is the falling factorial.\footnote{i.e. $m^{ \underline{k}} = m(m-1)\cdots(m-k+1)$.} The second-to-last inequality is by filling in the worst case for $w$ (and dropping the additional $d(i)/2$ in the denominator). The final inequality follows by
			suitably grouping terms in the falling factorial and noting that $\frac{x-k}{y-k} \leq \frac{x}{y}$ for all $0 \leq x \leq y$ and $0 \leq k < y$.
			
			Since for the Reed--Muller code the distance $d(i)$ is $\Omega(\sqrt{n(i)})$,
			this bound is negligible in $n(i)$.
			
			\item[$\frac{n(i)}{8} \leq w \leq \frac{n(i)-1}{2}$:] For these weights we
			will compare the total number of elements of $C_{i,1}$ with
			$\binom{n(i)}{w}$, and show that $\frac{ |C_{i,1}| }{ \binom{n(i)}{w} } = \negl(n(i))$.
			
			Let $h(p) = -p \log p -(1-p) \log(1-p)$ be the binary entropy function.
			The total number of elements of $C_{i,1}$ is $2^{ 2^{2i} } = 2^{ \frac{1}{2}(n(i)+1) }$. Without loss of generality, assume we look at $w = \frac{n(i)}{8}$,
			since the quantity we are computing, $\frac{|C_{i,1}|}{ \binom{n(i)}{w} }
			$, is monotonically decreasing for all $w$ in the range we are considering.
			
			As a rough bound, for our case of $ \alpha  = \frac{1}{8}$, we can use that $\binom{n(i)}{\alpha n(i)} \geq\frac{1}{O(\sqrt{n(i)})} 2^{h(\alpha) n(i)} $. (This bound follows from applying Stirling's approximation to the factorial formula for binomial coefficients.)
			Note that $h(1/8) \approx 0.54 > \frac{1}{2}$, and therefore the ratio
			is upper bounded by a function that is negligible in $n(i)$.
		\end{description}
	Combining the cases in the above analysis, we see that the maximum of $B_i^{\X}(w) / A_i^{\X}(w)$ over all $w \leq n(i)$ is upper bounded by the maximum of two negligible functions, which is itself negligible in $n(i)$.
		\qed
		
	\end{proof}
	
	\begin{lemma}[Ray-Chaudhuri--Wilson inequality \cite{RW75}]\label{lem:RayChaudhuriWilson}
		If $\mathcal{F}$ is a family of $k$-uniform $L$-intersecting subsets of a set of $n$ elements, where $\abs{L} = s$, then,
		\[
		\abs{\F} \leq \binom{n}{s} .
		\]
		Here $L = \{l_1, ..., l_s\}$ is a collection of allowed intersection sizes. Then, $L$-intersecting means that for all $i \neq j$, the intersection size
		$\abs{\F_i \cap \F_j} \in L$. 
	\end{lemma}

%%%%%%%%%%%%%%%%%%%%%%%%%%%%%%%%%%%%%%%%%%%%%%%%%%%%%%%%%%%%%%%
\section{Proof of Theorem~\ref{thm:parallel-key-reuse}}\label{appendix:proof-parallel-key-reuse}
%%%%%%%%%%%%%%%%%%%%%%%%%%%%%%%%%%%%%%%%%%%%%%%%%%%%%%%%%%%%%%%

We will state and prove security under parallel encryptions and sequential decryptions, in a setting that we described in Section~\ref{sec:parallel-encryptions}.

As a shorthand, define $\mathrm{Proj}_{acc}$ and $\mathrm{Proj}_{rej}$ as the quantum operations induced by projecting on, respectively, the accepting and rejecting outcomes of the decryption procedure.
That is, the accept/reject projectors on a message register $M$ are defined as
\[
\mathrm{Proj}^{\reg{M}}_{acc}: \rho \to \Bigl(\id^{\reg{M}} - \proj{\bot}^{\reg{M}} \Bigr) \rho \Bigl(\id^{\reg{M}} - \proj{\bot}^{\reg{M}} \Bigr)^\dagger
\]
and
\[
\mathrm{Proj}^{\reg{M}}_{rej}: \rho \to \proj{\bot}^{\reg{M}} \rho \Bigl( \proj{\bot}^{\reg{M}} \Bigr)^\dagger \,.
\]
The full statement of the theorem then is as follows.
\\\\
\mbox{{\bfseries Theorem~\ref{thm:parallel-key-reuse} (formal statement). \ }}
{\itshape
Let $(\Encrypt,\Decrypt)$ be an $\epsilon$-$\QCAR$-authenticating scheme resulting from Construction~\ref{con:qecc-to-skqes}, using a strong-purity-testing code $\{V_{k_0}\}_{\Key_0}$. Let $M_1, M_2$ denote the plaintext registers of the two messages, $C_1 = M_1T_1, C_2 = M_2T_2$ the corresponding ciphertext registers, and $R$ a side-information register. Let $\advA_1$, $\advA_2$ be arbitrary adversarial channels. Define the effective real channel, for keys $k_0, k_1, k_2$, as 
\[
\mathrm{Real}_{k_0,k_1,k_2} = \Decrypt^{\reg{C_2 \to M_2}}_{k_0,k_2} \circ \advA^{\reg{M_1,C_2,R}}_2 \circ \Decrypt^{\reg{C_1 \to M_1}}_{k_0,k_1} \circ \advA^{\reg{C_1,C_2,R}}_1 \circ \left(\Encrypt^{\reg{M_1 \to C_1}}_{k_0,k_1} \otimes \Encrypt^{\reg{M_2 \to C_2}}_{k_0,k_2} \right)\,.
\]

There exist a simulator $\simrej$ such that the key-recycling real channel

\begin{align}
\mathfrak{R} : \rho^{\reg{M_1M_2R}} \mapsto \mathbb{E}_{k_0, k_1, k_2} \Big[ & \Bigl( \mathrm{Proj}^{\reg{M_2}}_{acc} \circ \mathrm{Real}_{k_0,k_1,k_2} (\rho) \Bigr) \otimes \proj{k_0,k_2}\nonumber \\
\ \  + \ \  & \Bigl(\mathrm{Proj}^{\reg{M_2}}_{rej} \circ \mathrm{Real}_{k_0,k_1,k_2} (\rho) \Bigr) \otimes \proj{k_0} \Big]
\end{align}
is $2\epsilon$-close in diamond-norm distance to the ideal channel,
\begin{align}
\mathfrak{I} : \id^{\reg{M_2}} \otimes \simacc \otimes \mathbb{E}_{k_0,k_2} \left[ \proj{k_0,k_2} \right] \ \ + \ \ \proj{\bot}^{\reg{M_2}} \Tr_{M_2} \otimes \ \simrej \otimes \mathbb{E}_{k_0} \left[\proj{k_0}\right].
\end{align}
Here, $\simacc$ is as in Definition~\ref{def:QCAR}, and $\mathrm{Proj}_{acc}$ and $\mathrm{Proj}_{acc}$ are the quantum operations induced by projecting on the accepting and rejecting outcomes of the decryption procedure of the second qubit. 
}
\begin{proof}
We can represent the starting real channel pictorially as follows. The `key reveal box' at the end corresponds to the key-recycling term in the real channel, the content of which depends on whether the scheme accepts or rejects.

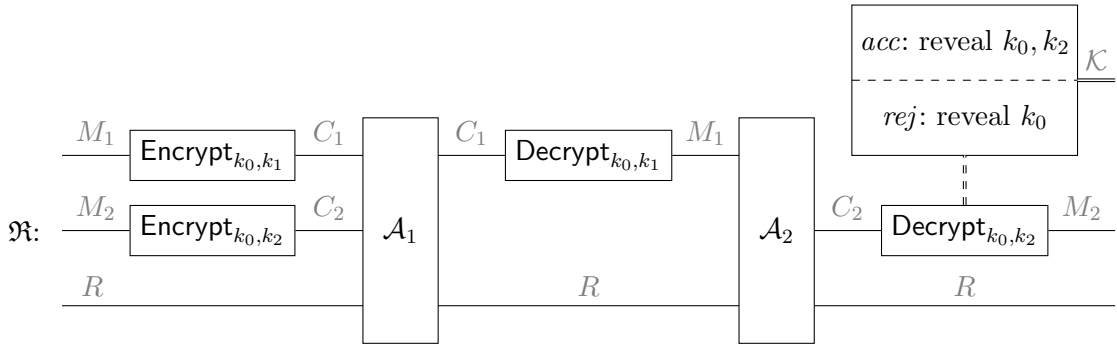
\begin{center}
\begin{tikzpicture}
\tikzstyle{operator} = [draw,fill=white,minimum size=1.5em] 

\node at (-0.5,5) {{$\mathfrak{R}$:}};

\coordinate (A1t) at (4, 6.5);
\coordinate (A1b) at (5, 3.5);

\coordinate (A2t) at (9, 6.5);
\coordinate (A2b) at (10, 3.5);

%\draw[help lines] (0,0) grid (15,10);
\node[operator] (Enc1) at (2, 6) {$\Encrypt_{k_0,k_1}$};
\node[operator] (Enc2) at (2, 5) {$\Encrypt_{k_0,k_2}$};

\node[operator] (Dec1) at (7, 6) {$\Decrypt_{k_0,k_1}$};
\node[operator] (Dec2) at (12, 5) {$\Decrypt_{k_0,k_2}$};

\draw[operator] (A1t) rectangle node {$\advA_1$} (A1b);

\draw  (Enc1) -- ++(-2,0) node[midway,above] {$\reg{M_1}$};
\draw  (Enc1) -- ++(2,0) node[midway,above] {$\reg{C_1}$};

\draw  (Enc2) -- ++(-2,0) node[midway,above] {$\reg{M_2}$};
\draw  (Enc2) -- ++(2,0) node[midway,above] {$\reg{C_2}$};

\draw  (Dec1) -- ++(-2,0) node[midway,above] {$\reg{C_1}$};
\draw  (Dec1) -- ++(2,0) node[midway,above] {$\reg{M_1}$};

\draw  (Dec2) -- ++(-2,0) node[midway,above] {$\reg{C_2}$};
\draw  (Dec2) -- ++(2,0) node[midway,above] {$\reg{M_2}$};

\draw (0, 4) -- node[pos=0.1,above] {$\reg{R}$} (4, 4);
\draw (5, 4) -- node[above] {$\reg{R}$} (9, 4);
\draw (10, 4) -- node[above] {$\reg{R}$} (14, 4);

\draw[operator] (A2t) rectangle node {$\advA_2$} (A2b);

% ''Conditional key-reveal box''
\draw[operator] (10.5, 8) rectangle (13.5, 6);
\draw[dashed] (10.5, 7) -- (13.5, 7);

\draw[draw=none] (10.5, 8) rectangle node{\emph{acc}: reveal $k_0,k_2$} (13.5, 7);
\draw[draw=none] (10.5, 7) rectangle node{\emph{rej}: reveal $k_0$} (13.5, 6);

\draw[double] (13.5, 7) -- node[above] {$\reg{\mathcal{K}}$} (14, 7);
\draw[double,dashed] (Dec2) -- ++(0, 1);
\end{tikzpicture}
\end{center}

First, note that the one-time pad key $k_1$ is picked completely at random, used only for the encryption and decryption of the first qubit. In particular, in the final situation as represented by $\mathfrak{R}$ we do not require key-recycling for $k_1$.

We can view applying a quantum-one time pad that uses a uniformly-random key as completely equivalent to performing a teleportation.
Therefore our real channel is the same, whether this key $k_2$ is picked randomly beforehand, or is the random outcome $\hat{k}_2$ of a teleportation measurement. We denote the channel where the key of the second qubit comes from teleportation by $\mathfrak{R_2}$. Because this rewriting doesn't change functionality, we have that $\dnorm{\mathfrak{R_2}-\mathfrak{R}} = 0$.

We depict the channel $\mathfrak{R_2}$ pictorially below, which is the easiest representation. For completeness we will also write out the channel symbolically.

Let $\Gamma^{\reg{E C_2 \to \hat{\mathcal{K}}_2} }$ be the quantum channel that performs a pairwise Bell measurement between the $n$ qubits in register $E$ and the $n$ qubits in $C_2$, and stores the $2n$ classical outcome bits as $\hat{\mathcal{K}}_2$. That is, if we define 
\begin{equation}
\ket{\Phi_k}_{\reg{EC_2}} = \left(\bigotimes_{i=1}^n \X^{k_i} \right)^{\reg{E}} \left(\bigotimes_{i=n+1}^{2n} \Z^{k_i} \right)^{\reg{E}} \bigotimes_{i=1}^{n} \ket{\Phi^+}^{\reg{(E)_i (C_2)_i} }
\end{equation}
to be the basis of all possible $n$-qubit Bell states between pairwise qubits of $E$ and $C_2$, the measurement channel is defined by the operation
\begin{equation}
\Gamma^{\reg{E C_2 \to \hat{\mathcal{K}}_2}} (\rho) = \sum_{\hat{k}} \proj{\hat{k}} \bra{\Phi_{\hat{k}}}^{\reg{EC_2}} \rho^{\reg{EC_2}} \ket{\Phi_{\hat{k}} }^{\reg{EC_2}}.
\end{equation}

Define the shorthand:
\begin{align*}
\mathrm{Real'}^{\reg{M_1 M_2 R \to M_2 R \hat{\mathcal{K}_2}} }_{k_0,k_1} = \,& \Decrypt^{\reg{C_2 \to M_2}}_{k_0, \hat{\mathcal{K}}_2 } \circ \advA^{\reg{M_1,C_2,R}}_2 \circ \Decrypt^{\reg{C_1 \to M_1}}_{k_0,k_1} \\
& \circ \left(\advA^{\reg{C_1,D,R}}_1 \otimes \Gamma^{\reg{E C_2 \to \hat{\mathcal{K}}_2}} \right) \circ \left(\Encrypt^{\reg{M_1 \to C_1}}_{k_0,k_1} \otimes \Encrypt^{\reg{M_2 \to C_2}}_{k_0,0} \otimes \proj{\Phi^+}^{\reg{DE}}\right) \,,
\end{align*}
where we slightly abuse notation in letting the final $\Decrypt$ use a copy of the contents of the $\hat{\mathcal{K}}_2$ register as a key.
The rewritten channel becomes
\begin{align}
\mathfrak{R_2} : \rho^{\reg{M_1M_2R}} \mapsto \mathbb{E}_{k_0, k_1} \Big[ &\left( \mathrm{Proj}^{\reg{M_2}}_{acc} \circ \mathrm{Real'}^{\reg{M_1 M_2 R \to M_2 R \hat{\mathcal{K}_2}} }_{k_0,k_1} (\rho) \right) \otimes \proj{k_0}\nonumber \\
\ \  + \ & \left(\Tr_{\hat{\mathcal{K}}_2} \mathrm{Proj}^{\reg{M_2}}_{rej} \circ \mathrm{Real'}^{\reg{M_1 M_2 R \to M_2 R \hat{\mathcal{K}}_2} }_{k_0,k_1} (\rho) \right) \otimes \proj{k_0} \Big] \,.
\end{align}
(In the accept branch, the $\hat{\mathcal{K}}_2$ register is output as part of the recycled key.)
%------------------------------------------------------------------
%Second picture: real channel after applying teleportation trick:\\
%------------------------------------------------------------------
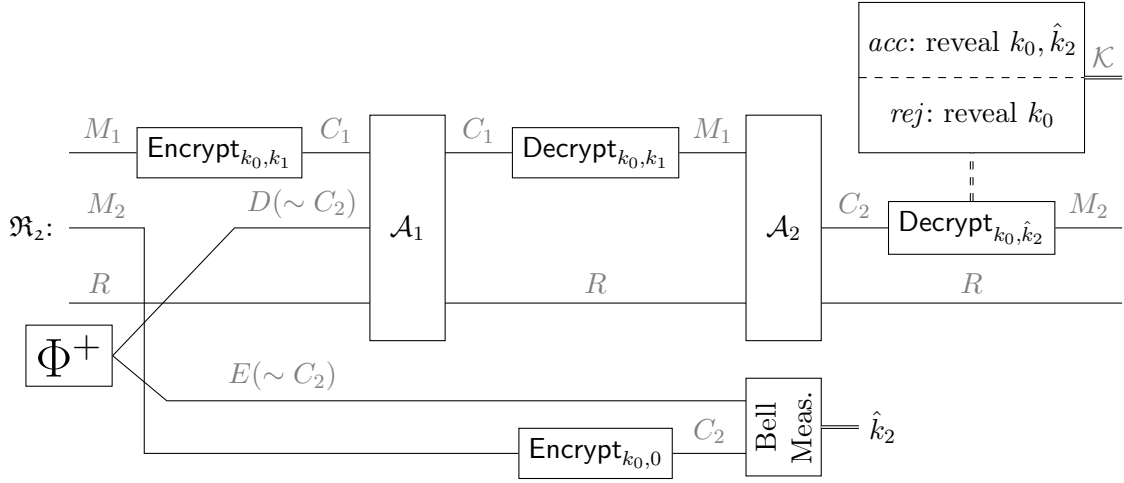
\begin{center}
\begin{tikzpicture}
\tikzstyle{operator} = [draw,fill=white,minimum size=1.5em] 

\node at (-0.5,5) {{$\mathfrak{R_2}$:}};

\coordinate (A1t) at (4, 6.5);
\coordinate (A1b) at (5, 3.5);

\coordinate (A2t) at (9, 6.5);
\coordinate (A2b) at (10, 3.5);

%\draw[help lines] (0,0) grid (15,10);
\node[operator] (Enc1) at (2, 6) {$\Encrypt_{k_0,k_1}$};
\node[operator] (Enc2) at (7, 2) {$\Encrypt_{k_0, 0}$};

\node[operator] (Dec1) at (7, 6) {$\Decrypt_{k_0,k_1}$};
\node[operator] (Dec2) at (12, 5) {$\Decrypt_{k_0,\hat{k}_2}$};

\draw[operator] (A1t) rectangle node {$\advA_1$} (A1b);

\node[operator] (Phi) at (0,3.3) {\huge{$\Phi^+$}};
\draw (Phi.east) -- (2.2,5) -- (4,5) node[midway,above] {$\reg{D (\sim C_2)}$};

\draw (Phi.east) -- (1.3,2.7) -- (9,2.7) node[pos=0.2,above] {$\reg{E(\sim C_2)}$};
\draw (Enc2.east) -- (9,2) node[midway,above] {$\reg{C_2}$};

%\draw (3,5)  -- ++(1,0) node[midway,above] {$\reg{C_2}$};
\draw (0,5)  -- (1,5) node[midway,above] {$\reg{M_2}$} |- (Enc2);

\draw  (Enc1) -- ++(-2,0) node[midway,above] {$\reg{M_1}$};
\draw  (Dec1) -- ++(-2,0) node[midway,above] {$\reg{C_1}$};
\draw  (Dec2) -- ++(-2,0) node[midway,above] {$\reg{C_2}$};

\draw  (Enc1) -- ++(2,0) node[midway,above] {$\reg{C_1}$};
\draw  (Dec1) -- ++(2,0) node[midway,above] {$\reg{M_1}$};
\draw  (Dec2) -- ++(2,0) node[midway,above] {$\reg{M_2}$};

\draw (0, 4) -- node[pos=0.1,above] {$\reg{R}$} (4, 4);
\draw (5, 4) -- node[above] {$\reg{R}$} (9, 4);
\draw (10, 4) -- node[above] {$\reg{R}$} (14, 4);

\draw[operator] (A2t) rectangle node {$\advA_2$} (A2b);

\draw[operator] (9, 3) rectangle node[rotate=90,align=center] {Bell\\Meas.} (10,1.7);
\draw [double] (10,2.35) -- ++(0.5,0) node[right] {$\hat{k}_2$};

% ''Conditional key-reveal box''
\draw[operator] (10.5, 8) rectangle (13.5, 6);
\draw[dashed] (10.5, 7) -- (13.5, 7);

\draw[draw=none] (10.5, 8) rectangle node{\emph{acc}: reveal $k_0,\hat{k}_2$} (13.5, 7);
\draw[draw=none] (10.5, 7) rectangle node{\emph{rej}: reveal $k_0$} (13.5, 6);

\draw[double] (13.5, 7) -- node[above] {$\reg{\mathcal{K}}$} (14, 7);
\draw[double,dashed] (Dec2) -- ++(0, 1);

\end{tikzpicture}
\end{center}

Our next step will be to apply the $\QCAR$ security of qubit 1. Note that we can shift the encryption of qubit 2, and the Bell measurement afterwards, to occur after qubit 1 is decrypted. This is possible since none of the registers $M_2, E, \hat{\mathcal{K}}_2$ are used by $\advA_1$ in the channel
$\mathfrak{R}_2$.

Now we view $\mathfrak{R}_2$ as a real channel for qubit 1, composed with some other quantum channel that uses $k_0$.
(Again, we will write the components of channel out, but the transformation will be clearest by looking at the diagram.)

We define $\mathfrak{S}$ as the real channel (as in the definition of $\QCAR$) of adversary $\advA_1$ acting on qubit 1. That is,

\begin{align}
\mathfrak{S}^{\reg{M_1 D R \to M_1 R \mathcal{K}_0 \mathcal{K}_1}}: \rho \to \expectation_{k_0,k_1}  \Bigl[ & \left(\mathrm{Proj}^{\reg{M_1}}_{acc} \otimes
\Decrypt^{\reg{M_1 \to C_1}}_{k_0,k_1} \circ
 \advA^{\reg{C_1,D,R}}_1 \circ \Encrypt^{\reg{M_1 \to C_1}}_{k_0,k_1} \right) (\rho) \otimes \proj{k_0,k_1}
 \nonumber \\ & +
\left( \mathrm{Proj}^{\reg{M_1}}_{rej} \otimes
\Decrypt^{\reg{M_1 \to C_1}}_{k_0,k_1} \circ
 \advA^{\reg{C_1,D,R}}_1 \circ \Encrypt^{\reg{M_1 \to C_1}}_{k_0,k_1} \right)(\rho) \otimes \proj{k_0}
 \Bigr].
\end{align}
Then by security of the scheme, we know that there exists a simulator $\simS^{\reg{D,R}}_1 = \simacc^1 + \simrej^1$ such that the corresponding ideal channel $\mathfrak{T}$ is $\epsilon$ close to the real channel in the diamond-norm. The ideal channel is given by
\begin{align}
\mathfrak{T}^{\reg{M_1 D R \to M_1 R \mathcal{K}_0 \mathcal{K}_1}} : \rho \to & \left(\id^{\reg{M_1}} \otimes \simacc^1\right) \left(\rho \otimes \mathbb{E}_{k_0,k_1} \left[ \proj{k_0,k_1} \right] \right) \nonumber \\
&+ \left(\proj{\bot}^{\reg{M_1}} \Tr_{M_1} \otimes \simrej^1 \right) \left( \rho \otimes \mathbb{E}_{k_0} \left[\proj{k_0}\right]  \right).
\end{align}

Now we note that it's possible to construct a quantum operation $\mathfrak{U}_{\mathcal{K}_0}$ so that we can write $\mathfrak{R_2}(\rho)$ as
\begin{equation}
\mathfrak{U}_{\mathcal{K}_0} \circ (\Tr_{K_1} \mathfrak{S}^{\reg{M_1 D R}} \otimes \id^{\reg{E,C_2}}) (\rho^{\reg{M_1,M_2,R}} \otimes \proj{\Phi^+}^{\reg{DE}}) \,.
\end{equation}
First the part of the key corresponding to $k_1$, which is only recycled in the accept case, is traced out. The operation $\mathfrak{U}_{\mathcal{K}_0}$ corresponds to everything happening after decryption of the first qubit: Use the recycled key $k_0$, which is output by $\mathfrak{S}$, to encode the second qubit, and then perform the teleportation step as before. The operation ends by executing $\advA_2$, decrypting, and recycling keys as required.

We construct the next channel $\mathfrak{R}_3 (\rho)$ by replacing the qubit-one real channel $\mathfrak{S}$ in this formulation by the ideal channel $\mathfrak{T}$. That is, the channel becomes
\begin{equation}
\mathfrak{U}_{\mathcal{K}_0} \circ (\Tr_{K_1} \mathfrak{T}^{\reg{M_1 D R}} \otimes \id^{\reg{E,C_2}}) (\rho^{\reg{M_1,M_2,R}} \otimes \proj{\Phi^+}^{\reg{DE}}) \,.
\end{equation}
The diamond-norm distance is constant under tensor product with the identity channel. The distance between any two channels is also monotonely non-increasing when applying a quantum operation afterwards, or when fixing a part of the input beforehand. Therefore, because $\QCAR$ security of the scheme implies
$\frac{1}{2}\dnorm{\mathfrak{S}-\mathfrak{T}} \leq \epsilon$, we have that $\frac{1}{2}\dnorm{\mathfrak{R}_3-\mathfrak{R}_2} \leq \epsilon$.

%------------------------------------------------------------------
%Third picture: apply security on qubit 1: (losing $\epsilon$). Here we use recycling, and use the fact that it's diamond norm.\\
%------------------------------------------------------------------
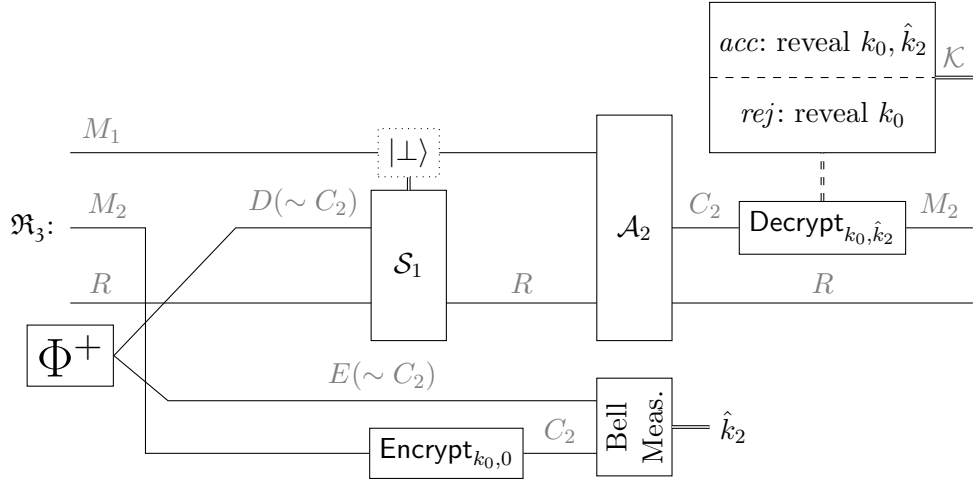
\begin{center}
\begin{tikzpicture}
\tikzstyle{operator} = [draw,fill=white,minimum size=1.5em] 

\node at (-0.5,5) {{$\mathfrak{R_3}$:}};

\coordinate (A1t) at (4, 5.5);
\coordinate (A1b) at (5, 3.5);

\coordinate (A2t) at (7, 6.5);
\coordinate (A2b) at (8, 3.5);

%\draw[help lines] (0,0) grid (15,10);
\node[operator] (Enc2) at (5, 2) {$\Encrypt_{k_0, 0}$};

\node[operator] (Dec2) at (10, 5) {$\Decrypt_{k_0,\hat{k}_2}$};

\draw[operator] (A1t) rectangle node {$\simS_1$} (A1b);

\node[operator] (Phi) at (0,3.3) {\huge{$\Phi^+$}};
\draw (Phi.east) -- (2.2,5) -- (4,5) node[midway,above] {$\reg{D (\sim C_2)}$};

\draw (Phi.east) -- (1.3,2.7) -- (7,2.7) node[midway,above] {$\reg{E(\sim C_2)}$};
\draw (Enc2.east) -- (7,2) node[midway,above] {$\reg{C_2}$};

%\draw (3,5)  -- ++(1,0) node[midway,above] {$\reg{C_2}$};
\draw (0,5)  -- (1,5) node[midway,above] {$\reg{M_2}$} |- (Enc2);

\draw  (0,6) -- (7,6) node[pos=0.06,above] {$\reg{M_1}$};
\draw  (Dec2) -- ++(-2,0) node[midway,above] {$\reg{C_2}$};

\draw  (Dec2) -- ++(2,0) node[midway,above] {$\reg{M_2}$};

\draw (0, 4) -- node[pos=0.1, above] {$\reg{R}$} (4, 4);
\draw (5, 4) -- node[above] {$\reg{R}$} (7, 4);
\draw (8, 4) -- node[above] {$\reg{R}$} (12, 4);

\draw[operator] (A2t) rectangle node {$\advA_2$} (A2b);

\draw[operator] (7, 3) rectangle node[rotate=90,align=center] {Bell\\Meas.} (8,1.7);
\draw [double] (8,2.35) -- ++(0.5,0) node[right] {$\hat{k}_2$};

% conditional erasure box
\node[operator,dotted] (cerase) at (4.5,6) {$\ket{\bot}$};
\draw[double] (4.5,5.5) -- (cerase);

% ''Conditional key-reveal box''
\draw[operator] (8.5, 8) rectangle (11.5, 6);
\draw[dashed] (8.5, 7) -- (11.5, 7);

\draw[draw=none] (8.5, 8) rectangle node{\emph{acc}: reveal $k_0,\hat{k}_2$} (11.5, 7);
\draw[draw=none] (8.5, 7) rectangle node{\emph{rej}: reveal $k_0$} (11.5, 6);

\draw[double] (11.5, 7) -- node[above] {$\reg{\mathcal{K}}$} (12, 7);
\draw[double,dashed] (Dec2) -- ++(0, 1);

\end{tikzpicture}
\end{center}
The dotted box in the picture represents the conditional replacement of $M_1$ by the state $\proj{\bot}$, in the reject case of the simulator $\simS_1$.

We have applied the security definition on the first qubit and effectively gotten rid of prior dependence of key $k_0$. Now, we can undo the teleportation rewrite and encrypt the second qubit at the start again. Represent this rewrite as $\mathfrak{R_4}$. Because the functionality is equivalent, we have that
$\frac{1}{2}\dnorm{\mathfrak{R_3} - \mathfrak{R_4}} = 0$.
%------------------------------------------------------------------
%Fourth picture: undo teleportation trick:\\
%------------------------------------------------------------------
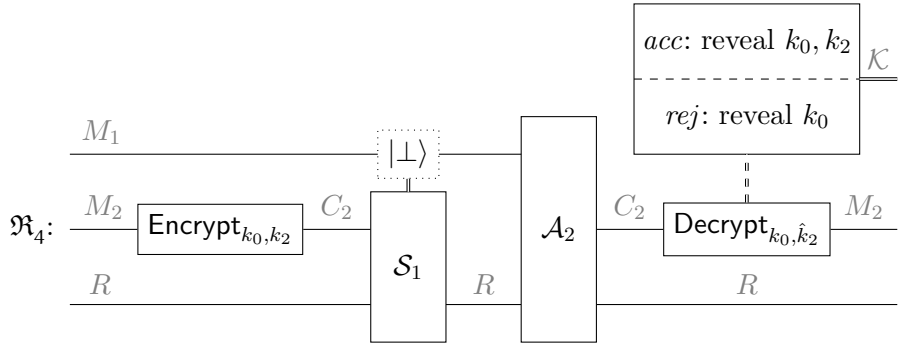
\begin{center}
\begin{tikzpicture}
\node at (-0.5,5) {{$\mathfrak{R_4}$:}};

\tikzstyle{operator} = [draw,fill=white,minimum size=1.5em] 

\coordinate (A1t) at (4, 5.5);
\coordinate (A1b) at (5, 3.5);

\coordinate (A2t) at (6, 6.5);
\coordinate (A2b) at (7, 3.5);

%\draw[help lines] (0,0) grid (15,10);

\node[operator] (Enc2) at (2, 5) {$\Encrypt_{k_0,k_2}$};

\draw  (Enc2) -- ++(-2,0) node[midway,above] {$\reg{M_2}$};
\draw  (Enc2) -- ++(2,0) node[midway,above] {$\reg{C_2}$};

\node[operator] (Dec2) at (9, 5) {$\Decrypt_{k_0,\hat{k}_2}$};

\draw[operator] (A1t) rectangle node {$\simS_1$} (A1b);

\draw  (0,6) -- (7,6) node[pos=0.06,above] {$\reg{M_1}$};
\draw  (Dec2) -- ++(-2,0) node[midway,above] {$\reg{C_2}$};

\draw  (Dec2) -- ++(2,0) node[midway,above] {$\reg{M_2}$};

\draw (0, 4) -- node[pos=0.1,above] {$\reg{R}$} ++(4, 0);
\draw (5, 4) -- node[above] {$\reg{R}$} ++(1, 0);
\draw (7, 4) -- node[above] {$\reg{R}$} ++(4, 0);

\draw[operator] (A2t) rectangle node {$\advA_2$} (A2b);

% ''Conditional key-reveal box''
\draw[operator] (7.5, 8) rectangle (10.5, 6);
\draw[dashed] (7.5, 7) -- (10.5, 7);

\draw[draw=none] (7.5, 8) rectangle node{\emph{acc}: reveal $k_0,k_2$} (10.5, 7);
\draw[draw=none] (7.5, 7) rectangle node{\emph{rej}: reveal $k_0$} (10.5, 6);

\draw[double] (10.5, 7) -- node[above] {$\reg{\mathcal{K}}$} (11, 7);
\draw[double,dashed] (Dec2) -- ++(0, 1);

% conditional erasure box
\node[operator,dotted] (cerase) at (4.5,6) {$\ket{\bot}$};
\draw[double] (4.5,5.5) -- (cerase);

\end{tikzpicture}
\end{center}

Finally we can apply the $\QCAR$ security again, for the channel depicted as $\mathfrak{R_4}$, with $\advA'_2 = \advA_2 \circ (\id^{\reg{M_1}} \otimes \simacc^1 + \proj{\bot}^{\reg{M_1}} \Tr_{M_1} \otimes \simrej^1)$ playing the role of the adversary in the security definition. It directly follows that $\frac{1}{2} \dnorm{\mathfrak{I} - \mathfrak{R_4}} \leq \epsilon$, i.e., the diamond-norm distance between $\mathfrak{R_4}$ and our aimed ideal channel $\mathfrak{I}$ is at most $\epsilon$. Also recall that $\frac{1}{2} \dnorm{\mathfrak{R_4} - \mathfrak{R}} \leq \epsilon$. Therefore, by triangle inequality, the original real channel $\mathfrak{R}$ is at most $2\epsilon$ from the ideal channel $\mathfrak{I}$.
\qed
\end{proof}

\end{document}